%% file: infovalab.tex
\documentclass[12pt]{article}
\usepackage{amssymb}
\usepackage{amsfonts}

\makeatletter
\input{abcommonheader}
\makeatother

\newtheorem{assumption}{Assumption}

\newtheorem{proposition}{Proposition}

\makeatletter
\newif\if@anonymize

\@anonymizefalse  %

\if@anonymize
    \newcommand{\anonymize}[1]{~}
    \renewcommand{\unfootnote}[1]{}
    \hypersetup{pdftitle={The Value of Information: A Puzzle},
        linkcolor=blue!50!black,citecolor=blue!50!black}

\else
    \newcommand{\anonymize}[1]{#1}
    \hypersetup{pdftitle={The Value of Information: A Puzzle},
        pdfauthor={Ohad Kadan and Asaf Manela},
        linkcolor=blue!50!black,citecolor=blue!50!black}
\fi

\usepackage[margin=1in]{geometry}
\makeatother
\begin{document}

\title{The Value of Information: A Puzzle
\unfootnote{We thank Yakov Amihud, Julie Fu, Nicolae G{\^a}rleanu, Xing Huang, P{\'e}ter Kondor, Tyler Muir, Christopher Schwarz, Taisiya Sikorskaya, Laura Veldkamp, S. "Vish" Viswanathan (discussant), seminar participants at LBS, LSE, Reichman University, and Washington University in St. Louis, and participants at the 2026 ASU-HKU Conference and 2026 FTG Bridging Theory and Empirical Research in Finance Conference for helpful comments. Claude Opus 5 and GPT 5 provided research assistance. 
E-mail: ohad.kadan@asu.edu or amanela@wustl.edu.}}
\author{\anonymize{Ohad Kadan \\ \normalsize{\textit{Arizona State University}}  \and  Asaf Manela\\ \normalsize{\textit{Washington University in St. Louis}}}}
\date{\normalsize August 2026}
\maketitle

\begin{abstract}

We show that the total value of information to informed traders can be measured by the covariation between price changes and order flow. This covariation captures noise trader losses, which equal informed trader gains when market making is competitive. We estimate the value of information using high frequency data on US equities at \$\EVwannualUSDm~million per year. The aggregate value of information is \meaninfovalpctofmarketcap\% of market cap, considerably lower than the \meanincrcostactivepctfrenchsample\% in fees investors pay searching for superior returns according to \textcite{french_presidential_2008}, which we extend through 2024. We discuss potential resolutions for these puzzling findings.

\end{abstract}
\medskip

\setcounter{page}{0} 
\thispagestyle{empty}  
\newgeometry{tmargin=1.5in,bmargin=1.5in,lmargin=1in,rmargin=1in}

\clearpage

\section{Introduction}

How much would you pay for information about tomorrow's closing stock price? What may at first sound like an arbitrage opportunity with infinite value, is in fact bounded in most markets due to price impact. Informed investors with foresight of a future stock price would act strategically and gradually trade to shade their superior information between the orders of other uninformed traders \parencite{kyle_continuous_1985}. The expected gain of the informed, which we call the value of information, provides an upper bound on the value added by active asset managers who promise, for a fee, to beat the market using their superior information. In what follows, we estimate the aggregate value of information for all stocks traded in the US and for all potentially informed investors combined. We find that it is surprisingly small compared to the aggregate amount of fees paid by investors to active asset managers \parencite[e.g.,][]{french_presidential_2008}.

Our approach to estimating the value of information relies on mild assumptions, ones satisfied by a broad class of dynamic models on strategic informed trading. These models include the \textcite{kyle_continuous_1985}, \textcite{back_insider_1992} and \textcite{back_information_2004} models of monopolistic informed traders, the \textcite{back_imperfect_2000} generalization to multiple informed investors with imperfect signals, and the \textcite{collin-dufresne_insider_2016} extension to stochastic noise trading. In such models, informed investors gain at the expense of noise or liquidity traders who lose, while competitive market makers break even on average. Noise trader losses on average equal informed trader gains when market making is competitive, and provide an upper bound on the value of information when market makers have market power \parencite{treynor_only_1971}. Rather than a hypothetical signal, mean noise trader losses measure the value of actual information in the market using the revealed preferences of market makers.

While measuring the average gains of the informed is often complicated, measuring the average losses of noise traders is surprisingly tractable. Noise traders lose every time they submit a market buy order $dZ$ and the price increases $dP$ due to price impact (or a bid-ask spread) set by market makers. They also lose if they sell and price decreases. It has long been known, therefore, that the total value of information, $\Omega$, equals noise traders' expected losses over the trading period $T$,
\begin{equation}
	\Omega=E\int_0^T dP dZ,\label{eq:EdPdZ}
\end{equation}
which is simply the covariation (formally the expected quadratic covariation) between price changes and noise trades.
Price changes $dP$ are clearly observable. However, the literature has so far struggled with the measurement of noise trading $dZ$ \parencite{manela_value_2014, peress_noise_2021}. If market makers cannot distinguish between informed trades and noise trades, how could an econometrician possibly do so?

Our key insight is that observing noise trading is unnecessary for measuring the covariation in Equation \eqref{eq:EdPdZ}. The reason is that in this class of models, the informed investor's trading strategy is locally orthogonal to price changes. 
Informed traders try to disguise their trades to avoid a strong price impact by smoothing their trades over time. As a result, informed trading $dX$ is ``smooth,'' while noise trading $dZ$ is ``jittery.''
Consequently, the covariation between price changes and changes in noise trading exactly equals the covariation between price changes and order flow $dY=dX+dZ$, as long as it is measured at a high enough frequency. This, in turn, implies that the total value of information \eqref{eq:EdPdZ} takes the form
\begin{equation}
	\Omega=E\int_0^T dP dY.\label{eq:EdPdY}
\end{equation}
Thus, measuring order flow $dY$, which is observable, suffices for quantifying the aggregate losses of uninformed noise traders, and therefore to measure the value of information.

Armed with this simple expression, we estimate the value of information daily with a discrete version of Equation \eqref{eq:EdPdY}, using
intraday data for US publicly-listed common stocks over the \sampleperiod%
~period. The idea is to measure an investor's willingness to pay for information on the average stock on the average trading day.
We find that the average daily value of
information is \$\EVwdailyUSDapprox, which annualizes to an average of 
\$\EVwannualUSDm~million for becoming informed on an average stock each day. 

The value of information we estimate is surprisingly small. The aggregate annualized value of information for all US stocks combined is only \meaninfovalpctofmarketcap\% of market capitalization. This value is the most active investors and asset managers can hope to gain by seeking superior information to forecast stock prices. 
For comparison, \textcite{french_presidential_2008} estimates that investors pay active asset managers about \meanincrcostactivepctfrenchsample\% of market capitalization in fees each year from 1980 to 2006 searching for superior returns. 
We extend French's sample and find that during 2003--2024, this cost declines modestly to \meanincrcostactivepctvoisample\% of market capitalization, which is still an order of magnitude larger than our estimate of the aggregate value of information. As \textcite{garleanu_efficiently_2018} shows, in a general equilibrium for both assets and asset management, the value of information equals management fees plus the cost of searching for good managers. We find it puzzling therefore that investors would pay so much more in fees than active asset management can possibly generate.

We initially met this finding with skepticism. We therefore devote considerable space to validation and robustness tests. In the time-series, we find that the value of information rises during turbulent times consistent with \textcite{kadan_liquidity_2025}. In the cross-section, we find that the value of information is higher for large-growth stocks, consistent with \textcite{farboodi_where_2021}, and higher for momentum stocks, consistent with \textcite{kadan_liquidity_2025}. On earnings announcement days, when we expect informed investors to have some of the most valuable information, we find that the value of information rises more than 3-fold. We also find that informed trading intensity as measured by \textcite{bogousslavsky_informed_2024} is higher when the value of information is higher. Moreover, the value of information is high when retail attention is high as measured by Google search volume \parencite{da_market_2025}.
In robustness tests, we find that the magnitude of the value of information hardly changes when we vary the frequency of observation, the way we sign trades as buys or sells, and if we include penny stocks.

In the last part of the paper we discuss potential resolutions to the value of information puzzle. We consider trends in the costs of active investing and other non-informational services provided by asset managers. We also consider strategies concentrating information acquisition on days with high value of information. We discuss the role of risk aversion, the role of partial information, the role of multiple informed investors as well as dynamic information acquisition. We discuss the role of non-competitive market makers and non-competitive asset managers. We generalize our results to discrete time, allowing for jumps, and relax the assumption that price changes are orthogonal to informed trading. We also consider the possibility that the value of information is correlated with the marginal value of wealth and consider cross-asset information. Finally, we discuss the role of behavioral investors.

Our paper is the first to estimate the value of information indirectly, by measuring the losses of noise traders, and the first to point out the conditions under which the (quadratic) covariation between price changes and order flow measures exactly this quantity. 
The generality and simplicity
of our measure allow us to reliably quantify the value of information 
and to document its puzzling magnitude relative to active asset management fees. 
For this purpose, it is paramount that it has easy to interpret units.\footnote{%
See \textcite{kyle_market_2016}, \textcite{kyle_market_2018}, %
\textcite{kyle_smooth_2018}, and \textcite{kyle_market_2019} 
for careful work on microstructure invariants and units of measurement.}
Economists have been interested in the value of information at least since \textcite{blackwell_equivalent_1953}, 
with recent theoretical work by \textcite{frankel_quantifying_2019} and \textcite{brooks_comparisons_2024}
increasingly aimed at measurement.
Here, we focus on the value of information to investors, and on a set of models that allow its measurement using trading data. The groundwork for our analysis lies in the classic contributions of
\textcite{grossman_impossibility_1980} and \textcite{kyle_continuous_1985},
who evaluate the value of information to investors within a rational‐expectations framework.
\textcite{manela_value_2014} estimates a structural model to value diffusing 
information about drug approvals and finds a small amount of noise trading.
\textcite{peress_noise_2021} proxy for noise trades with the trades of individual investors, whose weak correlation with fundamentals matches how these models define noise trading. \textcite{farboodi_where_2021, farboodi_valuing_2025} compute the value of information implied by a
structural noisy rational‐expectations model and show it is higher for large
growth firms. \textcite{bogousslavsky_informed_2024} develop a machine learning-based measure of informed trading intensity that captures the prevalence and timing of informed trading. Our approach requires no such proxies because the value of information depends on noise trades only through their covariation with prices, which order flow identifies.\footnote{%
The CARA–normal (mean–variance) noisy rational‐expectations framework
underpins a large literature on information choice; see
\textcite{veldkamp_information_2011} for a survey and
\textcite{malamud_decentralized_2017} and \textcite{davila_trading_2021} for
recent contributions. However, empirical implementation is challenging because
stock prices are neither stationary nor Gaussian, and recent studies suggest
that departures from CARA utility can be consequential
\parencites{savov_price_2014, malamud_noisy_2015}. Our estimate by contrast is agnostic to the functional form of the utility function.} 

\textcite{epstein_how_2014} and
\textcite{kadan_estimating_2019} examine the value of macroeconomic information
under \textcite{epstein_substitution_1989} preferences, which lets them
distinguish between the instrumental and psychic components of information. \textcite{croce_leading_2023} quantifies the relevance of heterogeneity in the timing of exposure of cash flows to aggregate shocks for returns.
\textcite{ai_identifying_2026} develops an asset market-based test for preference for the timing of resolution of uncertainty and apply it to FOMC announcements using options.
\textcite{davila_probability_2025} value marginal changes in probabilities, including changes in public and private information, as a standard asset-pricing formula with hypothetical cash flows derived from changes in the survival function.
 \textcite{kadan_liquidity_2025} show that the ratio of fundamental variance to price impact provides an approximation for the value of information in the \textcite{back_imperfect_2000} model with multiple differentially informed investors, and estimate it using high-frequency stock data. Here we provide a precise expression (rather than an approximation) for the value of information in a more general setting which incorporates this previous work as a special case, and use it to document its puzzling magnitude relative to active asset management fees.%

Our paper also relates to research on the informational efficiency of financial markets \parencite{fama_efficient_1970}. 
Recent work documents that price informativeness has risen since the 1960s, especially among large growth firms \parencite{bai_have_2016,farboodi_where_2021}.
\textcite{davila_identifying_2025} estimate informativeness by regressing prices on fundamentals and find it is greater for large stocks with high turnover, idiosyncratic volatility, institutional ownership, and analyst coverage. \textcite{kacperczyk_foreign_2021} report that higher foreign ownership boosts price informativeness, and \textcite{davila_volatility_2023} examine its relation to volatility. We contribute a simple measure of the value of information that incentivizes market participants to acquire private signals and, through trading, enhance capital allocation and managerial decisions \parencites{bond_real_2012, brogaard_economic_2019, goldstein_information_2023}. An interpretation of our finding that the value of information is modest, is that the equity market is highly efficient.

The paper proceeds as follows. 
Section \ref{sec:theory} provides the theoretical foundation
for our measure of the value of information. 
Section \ref{sec:estimation_and_validation} discusses its estimation.
Section \ref{sec:puzzle} discusses the value of information puzzle and its potential explanations.
Section \ref{sec:Conclusion} concludes. Robustness tests are in the Appendix.

\section{\label{sec:theory}Theory}

In this section, we use a generic version of continuous-time Kyle-Back models of trade with asymmetric information to develop a general closed-form expression for the value information. The main novel point we are making is that under standard assumptions, the exact functional form of the equilibrium is not needed to calculate the value of information, which takes a simple closed-form that is easy to estimate.  

\subsection{\label{model}Setup}

An asset is traded by three types of agents: informed traders, uninformed noise traders, and market makers.
Informed and uninformed submit their buy and sell orders, and the market makers set prices to clear the market.

The asset has a random fundamental (terminal) value $v$ with a continuous distribution function $F$. Similar to \textcite{back_insider_1992}, we do not make specific functional assumptions on $F$, which can accommodate a variety of different functional forms. The asset is traded continuously over a time interval $[0,T]$. One or more informed traders get some information about $v$ at time $t=0$. 
This information may be perfect as in \textcite{kyle_continuous_1985} or imperfect as in \textcite{back_imperfect_2000}.
Their cumulative orders at time $t$ are denoted by $X_t$, assumed to be a continuous It\^o process.

Noise traders demand $Z_t$ shares on net through time $t$. The cumulative noise trading process is given
by%
\begin{equation}
dZ_t =\sigma_{z}(t) dB_t , \label{eq:dZ_process}
\end{equation}%
where $B_t$ is a Wiener process and $\sigma_{z}(t)$ is the instantaneous volatility of noise trading, which may be stochastic as in \textcite{collin-dufresne_insider_2016}. Market makers cannot tell apart the informed and uninformed trades but observe the total order flow $Y_{t}=X_{t}+Z_{t}$. We denote the instantaneous variance of order flow by $\sigma^2_{y}$. Market makers set the price $P_t$ based on $Y_t$, such that the price process is a continuous It\^o process.\footnote{The assumption that $X_t$ and $P_t$ are continuous It\^o processes rules out jumps. See Section \ref{subsub:jumps} for a discussion of the effect of jumps on the value of information. Note that $Z_t$ is also a continuous It\^o process by virtue of the Wiener process $B_t$.} 

The noise traders' losses stem from prices moving against them when they trade.
For example, if they submit an order to buy $dZ_t>0$ shares and the price per share increases by $dP_t>0$, then they effectively pay a ``slippage'' cost of $dP_{t}dZ_{t}>0$. Given this, the total expected loss of noise traders is 
$
E\int_{0}^{T}dP_{t}dZ_{t}
$. Similar to \textcite{back_insider_1992}, unless mentioned otherwise, the operator $E$ denotes an expectation taken over $v$ as well as the stochastic variables in the integrand.\footnote{We follow \textcite{back_imperfect_2000} and use $\int_{0}^{T}dP_{t}dZ_{t}$ as an intuitive yet somewhat informal notation for the quadratic covariation between the It\^o processes $P_t$ and $Z_t$, often denoted $[P,Z]_T$.}

Our first key assumption is that market makers are competitive and thus earn zero profits on average. As a result, the value of information, which is the expected profit to informed traders (denoted by $\Omega$), must equal the expected loss to noise traders.

\begin{assumption}[Competitive Market Making]\label{as:competitive}
The total value of information to informed traders equals the expected loss to noise traders, 
\begin{equation}
\Omega=E\int_{0}^{T}dP_{t}dZ_{t}. \nonumber
\end{equation}
\end{assumption}

Assumption \ref{as:competitive} is standard in Kyle-Back models such as \textcite{kyle_continuous_1985}, \textcite{back_insider_1992}, \textcite{back_imperfect_2000}, and \textcite{collin-dufresne_insider_2016}. In such models, the value of information may be associated with model specifics such as the number of informed traders, the accuracy and correlation of their signals, and the distribution of the fundamental value, all of which are hard to measure empirically. Assumption \ref{as:competitive} says that in order to estimate the value of information what we need is to estimate the losses to noise traders, which, in turn, already capture all the above model parameters. Our fundamental insight is that, empirically, measuring the expected losses of noise traders is considerably simpler than measuring the expected profits of informed traders directly.

Our second assumption states that the informed trading process is stochastically uncorrelated with the price process.  

\begin{assumption}[Orthogonal Informed Trading]\label{as:orthogonal_dX}
Informed traders' strategy is orthogonal to the price process, i.e., $\int_{0}^{T}dX_{t}dP_{t}=0$.   
\end{assumption}

Formally, this means that the quadratic covariation of $X_t$ and $P_t$ is zero. Assumption \ref{as:orthogonal_dX} is a feature of all continuous-time Kyle-Back style models we are aware of including \textcite{kyle_continuous_1985}, \textcite{back_insider_1992}, \textcite{back_imperfect_2000}, \textcite{back_information_2004}, \textcite{collin-dufresne_insider_2016}, \textcite{banerjee_strategic_2020}, and \textcite{banerjee_dynamics_2022}. In these models, informed trading is locally deterministic and therefore has zero quadratic covariation with the price process.  

The intuition behind Assumption \ref{as:orthogonal_dX} is compelling and is now standard in Kyle-Back type models. Economically, the quadratic covariation, $\int_{0}^{T} dX_t dP_t,$ captures the cumulative losses to the informed traders from the price impact of their trades. This magnitude tends to be positive by the fact that market makers try to tease out the informed trades from noise trades through observed order flow.  \textcite{back_insider_1992} shows that informed traders optimally choose a smooth trading strategy that minimizes the cumulative price impact to their trades, thereby setting  $\int_{0}^{T} dX_t dP_t$ to zero, as stated in Assumption \ref{as:orthogonal_dX}. Thus, this assumption reflects that informed traders are trying to disguise their trades to avoid a strong price impact by smoothing their trades over time. 

As \textcite{back_insider_1992} notes, the condition $\int_{0}^{T} dX_t dP_t=0$ essentially means that  informed trades during small time increments of $\Delta t$ are of the order $\Delta t$. By contrast, noise trading is driven by a Brownian motion, and therefore the magnitude of noise-trading shocks during small time increments of $\Delta t$ is of order $\sqrt{\Delta t}$. Thus, informed trading is ``smooth,'' while noise trading is ``jittery.''  Variations in prices are then dominated by the jittery nature of noise trading, and the quadratic covariation of prices with informed trading vanishes. If informed traders were to add randomness to their trades in the form of a Brownian motion with increments of order $\sqrt{\Delta t}$ (like noise traders), their orders would contribute to the
diffusive part of prices. But such behavior is suboptimal in the Kyle-Back environment because it would render $\int_{0}^{T} dX_t dP_t$ being positive. The informed traders' optimal strategy in those models is therefore to spread trades smoothly over time, which
is exactly what forces the quadratic covariation of $X_t$ with $P_t$ to vanish.

\subsection{The value of information}

As noted, Assumptions \ref{as:competitive} and \ref{as:orthogonal_dX} are satisfied by many continuous-time Kyle-Back type models of trading under asymmetric information. While these models may differ in terms of the number of informed traders and the nature of private information, these two assumptions summarize standard and well-accepted features. 

Note that the stochastic loss by noise traders $\int_{0}^{T}dP_{t}dZ_{t}$ can be written as:
\begin{align}
  \int_{0}^{T}dP_{t}dZ_{t} &
  = \int_{0}^{T}dP_{t}(dY_{t} - dX_{t}) = \int_{0}^{T}dP_{t}dY_{t},
\end{align}
where the last equality is a consequence of Assumption \ref{as:orthogonal_dX}. It follows that the total expected loss to noise traders is given by $E \int_{0}^{T} dP_{t} dY_{t}$, 
which by Assumption \ref{as:competitive} is the value of information. We have established the following simple but powerful result.

\begin{proposition}[Value of information]\label{prop:voi_dPdY}
  Under Assumptions \ref{as:competitive} and \ref{as:orthogonal_dX}, the value of information over an interval of length $T$ is given by

\begin{equation}
\Omega =E\int_{0}^{T}dP_{t}dY_{t}.\nonumber
\end{equation}
\end{proposition}
\vspace{1cm}

The key in Proposition \ref{prop:voi_dPdY} is that we can replace noise trading changes $dZ_{t}$ with order flow $dY_{t}$ in the calculation of noise traders losses. While noise trading is not observable and cannot be easily inferred from data, order flow is observable and available at high frequency from standard databases, which facilitates our empirical analysis. It is important to note that noise trading and order flow are not equal to each other. The reason the substitution of $dZ_{t}$ with $dY_{t}$ in the proposition is valid is that under Assumption \ref{as:orthogonal_dX}, both have the same quadratic covariation with price changes. From an empirical perspective, Proposition \ref{prop:voi_dPdY} allows us to estimate the value of information by simply estimating the integral $\int_{0}^{T}dP_{t}dY_{t}$ using high frequency data, as an approximation to the continuous-time limit.

The strength of this result is that a simple estimable moment captures the value of information regardless of various model specific details. The reason is that the endgogenous optimal response of the market makers to order flow $dY$ is captured by price changes $dP$. When market makers perceive a high degree of adverse selection, they will set a high price impact, which will then translate into a high expected loss to the noise traders. 
This expected loss is then the value of information.
Whether this happens because of greater signal precision, a lower volatility of noise trading, or another reason, is irrelevant for our purposes. 
\subsection{A useful decomposition}
We next show that with two additional assumptions, the value of information can be decomposed into two intuitive components. It is important to emphasize that these additional assumptions are not necessary for our estimation of the value of information. We introduce them only to facilitate interpretation.

The first additional assumption imposes further restrictions on the trading strategies used by informed and noise traders.
\begin{assumption}[Diffusive trading strategies]\label{as:diffusive_trading_strategies}
The informed trading strategy $dX_{t}$ satisfies 
$$dX_t = \theta_t dt,$$ 
where $\theta_{t}$ is absolutely continuous. That is, $dX_t$ has no stochastic term.
Moreover, the noise trading process is given by%
$$dZ_t =\sigma_{z} dB_t .$$%
where $B_t$ is a Wiener process and $\sigma_{z}$ is a volatility parameter. 
\end{assumption}
The restriction on informed trading strategy is common in the literature. Indeed, in all Kyle-Back models we are aware of informed trading is locally deterministic. As for the restriction on noise trading, this assumption rules out the setting of \textcite{collin-dufresne_insider_2016} where $\sigma_z$ is itself a stochastic process. A direct implication of this assumption is that the quadratic variation of order flow is constant and equals the quadratic variation of noise trading, $(dY_t)^2=\sigma^2_z dt=\sigma^2_y dt.$

The second additional assumption is on the price process. In the classic Kyle-Back continuous-time models, prices are linear in order flow with Kyle's-lambda being the relevant slope. Critically, in those models, asset-specific order flow is the only reason for price changes. This strong assumption would suffice for our purposes. But, it is quite plausible, in fact, that other factors, orthogonal to the asset's order flow, may affect prices. For example, if a systematic risk factor, such as the market, changes in value, the price of individual stocks may change as well. To allow for that, let ${W_t}$ be a Wiener process uncorrelated with $B_t$. The idea is that ${W_t}$ represents publicly observable price-relevant information that is not captured in the order flow. We assume that prices respond to both order flow, ${Y_t}$, and to ${W_t}$ in a linear fashion:

\begin{assumption}[Linear prices]\label{as:linear_dP}
Order flow $Y_{t}$ affects price $P_{t}$ according to
\begin{equation}
dP_t=\lambda_{t} dY_t + \sigma_{w} dW_t, \label{eq:dPdY}
\end{equation}
for some price impact $\lambda_{t}>0$, coefficient $\sigma_{w}$, and a Wiener process $W_t$ which is orthogonal (i.e., has zero quadratic covariation) to $X_t$ and $B_t$.
\end{assumption}

Note that in the traditional Kyle-Back type models, price changes are assumed to be perfectly correlated with order flow and price impact is non stochastic. That is $\sigma_{w}$ is set to be identically zero and $\lambda_{t}$ is a function of time only. Therefore, Assumption \ref{as:linear_dP}, which allows for $\lambda_t$ to be stochastic and for additional sources of variation in prices, incorporates these settings as a special case. Because we abstract from many model specific details, our setup allows for other public information $W_t$ to arrive in the market and affect the asset's price as in \textcite{back_identifying_2018}. This is particularly useful from an empirical perspective, because, in reality, prices are not perfectly co-linear with order flow. 

The next proposition establishes that under Assumptions \ref{as:competitive}--\ref{as:linear_dP}, the value of information can be decomposed into a particularly intuitive form. 

\begin{proposition}[Value of information decomposition]\label{prop:voi_Elambda_sigmay}
  Under Assumptions \ref{as:competitive}--\ref{as:linear_dP},  the value of information
over an interval of length $T$ is the product of mean price impact $E(\bar{\lambda}_{t})=E(\frac{1}{T}\int_{0}^{T}\lambda_{t}dt)$ and order flow variance $\sigma_y^2T$:
\begin{equation}
\Omega = E(\bar{\lambda}_{t})\sigma_{y}^{2}T.\label{eq:Omega_prop_1}
\end{equation}

\begin{proof}
Using Assumption \ref{as:linear_dP}, 
\begin{align}
\int_{0}^{T}dP_{t}dY_{t} & =\int_{0}^{T}\lambda_{t}\left(dY_{t}\right)\left(dY_{t}\right)+\int_{0}^{T}\sigma_{w}\left(dW_{t}\right)\left(dY_{t}\right)\nonumber\\
 & =\int_{0}^{T}\lambda_{t}\left(dY_{t}\right)^2+\int_{0}^{T}\sigma_{w}\left(dW_{t}\right)\left(dX_{t} + dZ_t\right)\nonumber\\
 & =\int_{0}^{T}\lambda_{t}\sigma_{y}^{2}dt,
\end{align}
where the last equality follows from the fact that $W_t$ is orthogonal to $X_t$ and $Z_t$.%

By Proposition \ref{prop:voi_dPdY}, 

\begin{align}
\Omega & =E\left[\int_{0}^{T}dP_{t}dY_{t}\right]\nonumber \\
&=E\left[\int_{0}^{T}\lambda_{t}\sigma_{y}^{2}dt\right]\nonumber\\
 & =E\left[\frac{1}{T}\int_{0}^{T}\lambda_{t}dt\right]\sigma_{y}^{2} T\nonumber \\
 & =E(\bar{\lambda}_{t})\sigma_{y}^{2}T,
\end{align}
as stated.
\end{proof}

\end{proposition}

This result is intuitive. The value of information equals the expected aggregated losses of noise traders over time. By Proposition \ref{prop:voi_dPdY}, these losses are the time-aggregated product of price changes and the order flow. Proposition \ref{prop:voi_Elambda_sigmay} establishes that under linear pricing, the former aggregates to expected price impact, $E\left[\lambda_{t}\right]$, while the latter aggregates to order flow variance over the interval of length $T$, $\sigma_{y}^{2}T$. In our empirical analysis, we rely on this decomposition to gain economic intuition about variations in the value of information.

\subsection{Discussion and connection to earlier work}

Our approach in this paper is different from prior work in that we abstract from many detailed assumptions on fundamentals and instead, focus on standard features of Kyle-Back models that are crucial for the estimation of the value of information. In doing so, we are gaining simplicity and generality but at the cost of losing some details. Prior work such as \textcite{back_imperfect_2000}, delved very deeply into the dynamics of informed trading, yielding important insights on how informed traders spread their trades over time. Our purpose in this paper is quite different because our sole focus is on the estimation of the value of information. As it turns out, for this purpose, the detailed dynamics of informed trading are not needed because the value of information to informed traders is simply the expected losses to noise traders (Assumption \ref{as:competitive}). Thus, once one can estimate these losses, the exact nature of the dynamics of informed traders becomes unnecessary for the calculation of the value of information. As a result, one can avoid making assumptions on the value distribution, number of informed traders, the nature and precision of the information they possess, and the underlying correlation structure. 

In our setting, the challenge then becomes to estimate the losses of noise traders, which by itself is not observable. However, the quadratic covariation of noise trading with price changes is equal to that of total order flow -- a consequence of Assumption \ref{as:orthogonal_dX}. This, in turn, allows us to use variations in total order flow---which is empirically observable---to estimate the losses of noise traders. 

In \textcite{kadan_liquidity_2025}, we take a different approach to estimating the value of information. We derive approximate bounds on the value of information in one specific Kyle-type model -- that of \textcite{back_imperfect_2000}. We then show that in that setting, the value of information is approximately equal to the ratio of fundamental variance $\sigma_v^2$ to price impact at time zero -- $\lambda(0)$. The setting offered in this paper incorporates \textcite{kadan_liquidity_2025} as a special case. A major improvement with the approach provided in this paper is that we no longer rely on an approximation derived from lower and upper bounds on the value of information. Instead, the expressions obtained in Propositions \ref{prop:voi_dPdY} and \ref{prop:voi_Elambda_sigmay} are exact.

Another improvement relative to much of the Kyle-Back literature is that we do not require an exact linear relation between price changes and order flow (perfect correlation). Instead, we allow for price changes to be associated with other observable sources of information. This is particularly important from an empirical perspective---in our sample, for individual stocks, the mean correlation between prices changes and order flow is only about $0.3$.%

\subsection{Examples}
In this section, we provide two solved examples to demonstrate how the expressions in Propositions \ref{prop:voi_dPdY} and \ref{prop:voi_Elambda_sigmay} generalize known results. 

First, consider \textcite{kyle_continuous_1985}. His continuous auction (Section 5) studies a special case of our model in which $T=1$, $v$ is normally distributed, noise trading has no stochastic volatility, and there is one informed trader who obtains a perfect signal on the asset's value. Kyle shows that in this model, the value of information is given by $\Omega=\sigma_v\sigma_z$ -- the product of value volatility and noise trading volatility (page 1330). 

To see why this is a special case of our results, note that in Kyle's equilibrium, $\lambda_t$ is non-stochastic and constant over time, given by (Theorem 3) $\lambda=\sigma_v/\sigma_z$. Plugging this back in our Proposition \ref{prop:voi_Elambda_sigmay} gives 
\begin{equation}
\Omega = E(\bar{\lambda}_{t})\sigma_{y}^{2}={\lambda}\sigma_{y}^{2}=(\sigma_v/\sigma_z)\sigma_y^2. \label{eq:Kyle}
\end{equation}

Now, in Kyle's equilibrium, the informed trader's strategy is locally deterministic (no martingale component) and, therefore, $\sigma^2_y=\sigma_z^2$. Plugging back in (\ref{eq:Kyle}) gives $\Omega=\sigma_v\sigma_z$, coinciding with Kyle's expression.  

As a second example, consider \textcite{back_imperfect_2000}. They study a continuous-time model in which $T=1$, $v$ is normally distributed, $\sigma^2_z=1$, and noise trading has no stochastic volatility. In their model, there are multiple informed traders who obtain an imperfect signal on the asset's value. In their equilibrium (Theorem 1), $\lambda_t$ is locally deterministic but varies over time. Applying this special case to Proposition \ref{prop:voi_Elambda_sigmay} gives
\begin{equation}
\Omega = \bar{\lambda}_{t}\sigma_{y}^{2}=\sigma^2_y\int_0^1\lambda_tdt. \label{eq: voi_bcw}
\end{equation}
Now, noting that equilibrium informed trading in \textcite{back_imperfect_2000} is locally deterministic, their model results in $\sigma^2_y=\sigma_z^2=1$. Plugging back in (\ref{eq: voi_bcw}) gives 
\begin{equation}
\Omega = \int_0^1\lambda_tdt, 
\end{equation}
which coincides with their expression for the value of information in Corollary 3.

\section{\label{sec:estimation_and_validation}Estimation and validation}

We next explain how we estimate the value of information and validate our measure.

\subsection{\label{sub:estimation}Data and estimation}

We estimate daily values of information for each stock in the NYSE TAQ database,
which covers all US publicly traded stocks. 
For regulatory reasons, TAQ data is highly comprehensive and includes even off-market trades that occur in dark pools, or that were routed
to specific market markers in exchange for payment for order flow. 
Our sample includes \nstocks~common stocks over the \dailysampleperiod~sample
period, which we could match by trading symbol and date to CRSP/Compustat.
Following \textcite{amihud_illiquidity_2002} we drop all stocks with a
previous-day closing stock price smaller than \$5 to avoid market
microstructure effects.\footnote{In the Appendix we show how 
our estimates change if we keep penny stocks.} 

Empirically, we interpret the daily value of information for a particular
stock as the dollar amount informed investors would be willing to pay at the close of the
trading day for a signal about the closing price at the end of the next trading day. Our approach captures signal precision indirectly through the losses of noise traders. All else equal, a more precise signal means higher price impact and therefore higher value of information.

Our main estimator for the value of information relies on a discrete approximation of the integral in Proposition 
\ref{prop:voi_dPdY}. Let $P_{jt\tau}$ denote the price of asset 
$j$ on day $t$, at time $\tau\in[0,T]$, where $T$ denotes the time in years that passes on day $t$. We observe equidistant observations
over time intervals $h=\mbox{1 minute}$, so we estimate the sum at
time $T$ based on $N+1$ discrete observations recorded at times $\tau_0$, $%
\tau_1$, \ldots, $\tau_N=N h=T$ (We vary $h$ in robustness analysis). We annualize the value
of information by setting $T=1/252$. Let $\Delta P_{jti}=P_{jt\tau_i}-P_{jt\tau_{i-1}}$ be the change in the price of asset $j$ on day $t$ over the interval $%
i=1,\ldots,N$. Similarly, let $\Delta Y_{jti}=Y_{jt\tau_i}-Y_{jt\tau_{i-1}}$ be the change in cumulative signed share order flow, or simply net order flow over the same interval.
For each stock $j$ on date $t$, we estimate a daily annualized value of information, 
\begin{equation}
    \hat\Omega_{jt} = \frac{1}{T} \sum_{i=1}^N \Delta P_{jti} \Delta Y_{jti}.\label{eq:Omega_hat}
\end{equation}

For stationarity, we measure the value of information in two main ways. The first is in millions of dollars inflation-adjusted using the CPI to \cpibaseline values. The second is as a percentage of aggregate stock market capitalization. The first way is sometimes easier to interpret, while the second facilitates comparison with fees estimates such as in \textcite{french_presidential_2008}.

Several approaches to signing TAQ trades as buys ($+1$) or
sells ($-1$) have been proposed by prior work. Our baseline estimates are
based on the \textcite{chakrabarty_trade_2007} approach, using the
algorithms developed by \textcite{holden_liquidity_2014}. In the
Appendix we show that our qualitative and quantitative conclusions are not
sensitive to this choice.

We rely on the decomposition in Proposition \ref{prop:voi_Elambda_sigmay} to gain intuition for the source of variation in the value of information over time and across firms. The decomposition allows us to ask
if higher value is due to higher price impact or order flow variance.
To get a decomposition that does not depend on an arbitrary number of shares outstanding,
we focus on annualized dollar measures, and decompose the value of information as follows.
Let $C$ be a scaling constant used to convert a dollar amount to millions of real dollars.
Let $\tilde{\lambda} = \frac{E\left[\lambda_{t}\right]}{P_0/C}$ measure the effect of a million dollar buy order on returns.
Let $\tilde{\sigma}_y = \sigma_y P_0/C$ measure the annualized volatility of order flow in millions dollars.
Then, the annualized log value of information in real millions of dollars can be written,

\begin{equation}
\log\left(\Omega/C\right) = \log\tilde{\lambda} + \log\tilde{\sigma}_y^2. \label{eq:logvoi_decomposition}
\end{equation}
Of course, when we use this decomposition, we restrict the sample to observations with a positive price impact coefficient.

Consistent with Assumption \ref{as:linear_dP}, we estimate price impact $\lambda_{jt}$ using a regression of 1-minute log returns on
contemporaneous share order flow, 
\begin{equation}
r_{jti} \equiv \log \left( P_{jt\tau_{i}} / P_{jt\tau_{i-1}} \right) = \hat{\lambda}_{jt}\Delta Y_{jti}+w_{jti}.
\label{eq:lambda_regression}
\end{equation}
We note that while in theory $\lambda$ is positive, it could be negative in
any finite sample. In addition, $\hat{\lambda}_{jt}$ may be subject to selection bias if traders use price-dependent strategies and cancel costly orders \parencite{obizhaeva_selection_2011}. A large literature proposes other impact measures \parencite{holden_liquidity_2014}. We favor our measure because it aligns closely with the model used to interpret the data, making the value-of-information units straightforward.
Finally, we use the closing price of the preceding trading day, $%
P_{jt-1\tau_N}$ as the initial asset price $P_0$.

\subsection{Estimated information values and other summary statistics}

\begin{table}[tbp]
\caption{Information values and other summary statistics}
\label{tbl:summarystats} 
\par
\begin{center}
{\footnotesize \input{tables/summarystats_\baseapproach\basedir_%
\startyyyymm_\lastyyyymm.tex} }
\end{center}
\par
{\footnotesize Notes: Daily common stocks panel from NYSE TAQ, CRSP, and Compustat, \sampleperiod%
. \infoval \infovallambdasigmay \inflation \vofp \prcimp \size \momret \beme \earnings \winsor
}
\end{table}

Table \ref{tbl:summarystats} reports sample moments of the value of information 
as well as summary statistics for the sampled firms. 
Our estimates refer to the annualized value of acquiring information that is realized over a single day. Hence, the reported values represent the willingness to pay for receiving daily information on a stock for an entire year. The average annualized value of information in our sample is \$\EVwannualUSDm~million (\$\EVwdailyUSDapprox~per day), with substantial dispersion: the standard deviation is about \$\StdclusVwannualUSDm~million.

Alongside our main estimator \eqref{eq:Omega_hat}, derived from Proposition \ref{prop:voi_dPdY}, we also present estimates based on the decomposition in Proposition \ref{prop:voi_Elambda_sigmay}. In that case, we first estimate $\lambda$ and $\sigma_y^2$ separately for each stock–day and then take their product. The two sets of estimates are virtually identical, reinforcing that the decomposition closely approximates the general value of information.

The median price impact ( $\lambda$ divided by the stock price) is \medianpriceimpact, implying that \$1 million of order flow moves prices by \medianpriceimpactpct~percent.\footnote{Our median price‐impact estimate aligns with \textcite[Table II]{hasbrouck_trading_2009}, where a \$1 million order shifts the log price by 0.7\%. %
} About \negativeVwpct~percent of stock-day observations yield a negative information value estimate. We include these when reporting information value in levels, but exclude them when reporting results in logs below. The median firm in our sample has a market capitalization of \$\medianmktcapUSDb~billion, a book–to–market ratio of \medianbemepct~percent, and momentum indicating a \medianmompct~percent appreciation over the prior 11 months. Firms reported earnings on roughly \meanernpct~percent of sample days.

\subsection{\label{sub:infoval_time_series}The value of information over time}

We next turn to validating our measure by focusing on times when we expect the value of information to be particularly high, and by comparing it with established results.

\begin{figure}
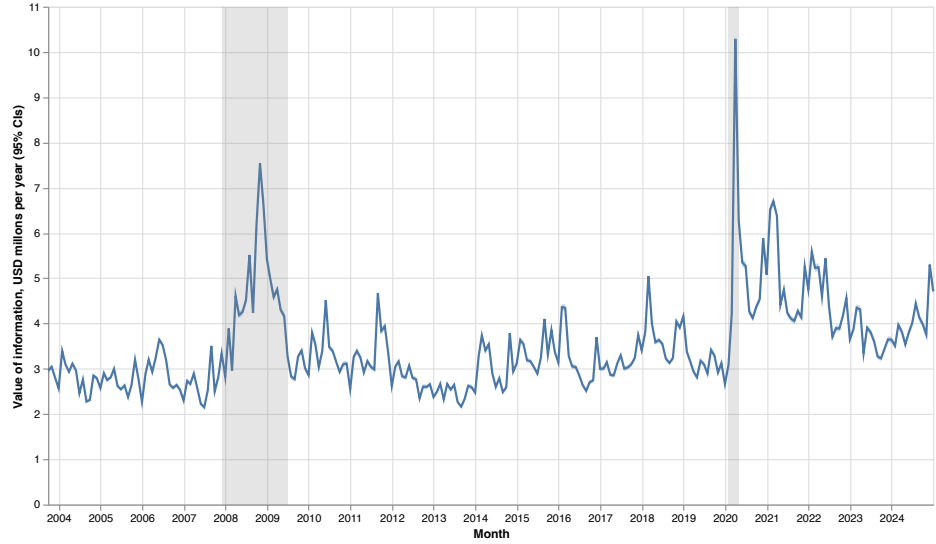

\caption{Value of information over time}
\label{fig:infoval_monthterms}
\begin{center}
\begin{subfigure}[t]{0.75\columnwidth}
\includegraphics[width=1\columnwidth]{figures/infoval_month_%
\baseapproach\basedir_\startyyyymm_\lastyyyymm.pdf}
\caption{\label{fig:infoval_month_levels}Value of information in levels}
\end{subfigure}
\par
\begin{subfigure}[t]{0.75\columnwidth} 
\includegraphics[width=1\columnwidth]{figures/lninfoval_monthterms_%
\baseapproach\basedir_\startyyyymm_\lastyyyymm.pdf}
\caption{\label{fig:infoval_month_logs}Value of information in logs and its components}
\end{subfigure}
\end{center}
\par
{\footnotesize Notes: The solid line is the monthly information value
averaged over stocks and days surrounded by its 95 percent confidence
interval, based on time-clustered variance estimates. \infoval \loginfovalterms \shades
}
\end{figure}

Figure \ref{fig:infoval_monthterms} traces the evolution of the value of information over time. Panel (a) plots the level of $\Omega$, averaged across stocks and days by month. Panel (b) shows its log decomposition into log price impact and log order flow volatility, which add up to log value of information. Recessions are shaded. The figure indicates that in turbulent periods—especially recessions—order flow volatility rises and liquidity deteriorates, as seen in higher price impact. These result in higher values of information, consistent with the findings of \textcite{kadan_liquidity_2025}. This pattern is most pronounced in the 2007--2009 financial crisis and the 2020 Covid-19 episode, where both order flow volatility and illiquidity surge, producing a spike in the value of information.

\subsection{\label{sub:infoval_cross_section}The value of information across firms}

Previous work has found that the value of information is higher for large (high market cap) and growth (low book-to-market) stocks \parencites{farboodi_where_2021} and for momentum (high recent returns) stocks \parencites{kadan_liquidity_2025}. We confirm that this pattern shows up when using our new measure with less noise.

\begin{table}
\caption{Value of information and stock characteristics}
\label{tbl:reg_on_characteristics}
\par
{\footnotesize 
\begin{subtable}[t]{1\columnwidth}
\input{tables/reg_\baseapproach\basedir_\startyyyymm_\lastyyyymm_lndPdYw_on_lnwtw_bemew_momretw.tex}
\caption{\label{tbl:lndPdYw_on_lnwtw_bemew_momretw}Information values are higher for large, growth and momentum stocks}
\end{subtable}
\begin{subtable}[t]{1\columnwidth}
\input{tables/reg_\baseapproach\basedir_\startyyyymm_\lastyyyymm_lnrvofpw_lnlampw_on_lnwtw_bemew_momretw.tex}
\caption{\label{tbl:lnrvofpw_lnlampw_on_lnwtw_bemew_momretw}Large-momentum stocks exhibit noisier order flow and greater liquidity}
\end{subtable}
}
\par
{\footnotesize Notes: Panel (a) shows regressions of the log value of
information and its components on stock characteristics. \loginfovalregs \size %
\momret \beme \neglambda \stderror \stars
}
\end{table}

Table \ref{tbl:reg_on_characteristics} studies the extent to which the value
of information can be explained by commonly studied characteristics: size,
book-to-market, and momentum. Panel \ref%
{tbl:lndPdYw_on_lnwtw_bemew_momretw} shows that large, growth and momentum
stocks have higher values of information. Unlike book-to-market, which is
subsumed by the inclusion of both stock and day fixed effects, large stocks and 
stocks that appreciated over the preceding year (high momentum stocks) have higher
values of information regardless of the regression specification.

Contrasting the R-squares to the Within-R-squares, we see that these characteristics explain a substantial share of the variation in information values, even after absorbing stock and day fixed effects.

Panel \ref{tbl:lnrvofpw_lnlampw_on_lnwtw_bemew_momretw} investigates the underlying 
reason for this variation using the decomposition of the value of information into price impact and order flow volatility (Proposition \ref{prop:voi_Elambda_sigmay}). It shows that large-momentum stocks exhibit noisier order flow and more liquid prices.

\subsection{\label{sub:infoval_earnings}The value of information around earning announcement days}

We next turn to studying how the value of information changes on dates in
which major pieces of information are released. Intuitively, a strategic
trader should be willing to pay a higher amount to learn the end-of-day
stock price at the beginning of days in which information is scheduled to be
released. To this end, we focus our attention on the most prominent
information release dates in the life of a firm---earnings announcement
dates.

We collect earnings announcement dates from the Compustat quarterly
database. Given that firms often report earnings after hours, we follow %
\textcite{engelberg_anomalies_2018} and define the earnings announcement
date as the date in which the trading volume in the firm scaled by market
volume is the highest among the official earnings date and the days just
before/after it.

\begin{figure}
\caption{Value of information rises when firms report earnings}
\label{fig:infoval_earntime}
\begin{center}
\includegraphics[width=0.9\columnwidth]{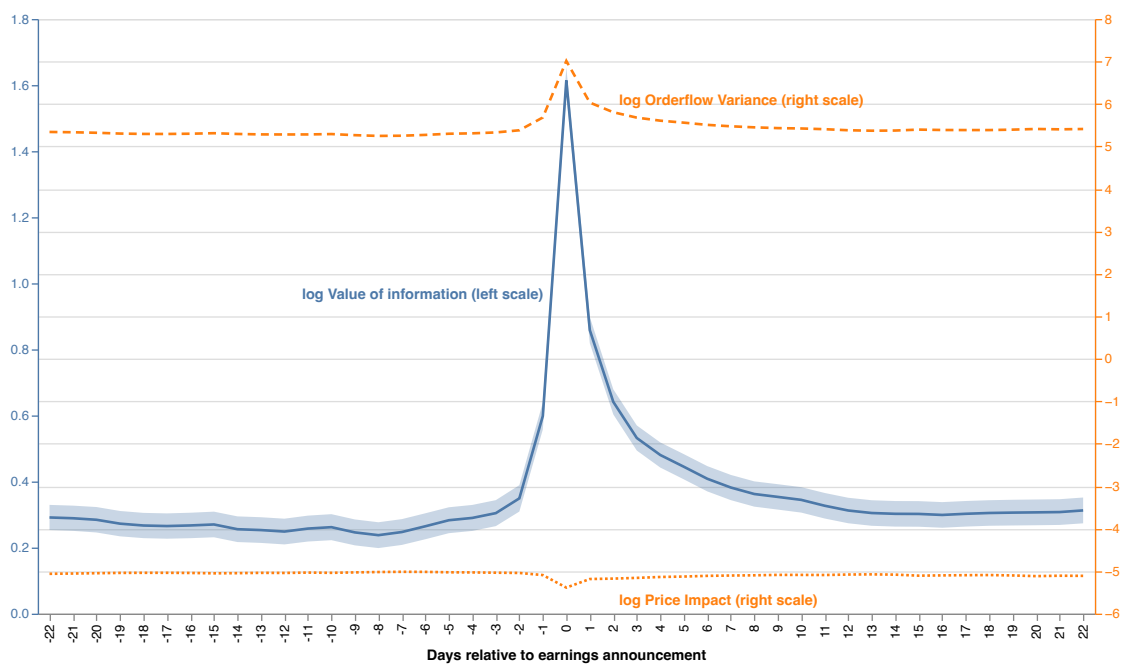}
\end{center}
\par
{\footnotesize Notes: \dailysolid \infoval \loginfovalterms
}
\end{figure}

Figure \ref{fig:infoval_earntime} plots the value of information averaged
across stocks and dates during a 45 trading day window around the earnings
announcement date (day 0). The figure shows a dramatic increase in the value
of information on days in which earnings are announced. This increase stems
primarily from a spike in volatility on these days whereas liquidity appears
to be actually mildly improving.

Table \ref{tbl:lndPdYw_on_ead_lnwtw_bemew_momretw} provides more formal
evidence supporting this conclusion. We regress the log of the value of
information for each firm and day during our sample on firm size,
book-to-market, momentum, date and stock fixed effects as well as an indicator
variable for earning-announcement dates. The coefficient of the earnings
indicator is positive and statistically significant in all specifications,
indicating that the value of information rises on earning-announcement days.
From an economic perspective, this effect is quite large, and implies that
the value of information increases more than 3-fold on the mean earnings
release ($e^{\erneffectonvaluebothfe} \approx \experneffectonvaluebothfe$).

\begin{table}
\caption{Information about earnings is highly valuable}
\label{tbl:reg_ead}
\par
{\footnotesize 
\begin{subtable}[t]{1\columnwidth}
\input{tables/reg_\baseapproach\basedir_\startyyyymm_\lastyyyymm_lndPdYw_on_ead_lnwtw_bemew_momretw.tex}
\caption{\label{tbl:lndPdYw_on_ead_lnwtw_bemew_momretw}Value of information is higher on earnings announcement days}
\end{subtable}
\begin{subtable}[t]{1\columnwidth}
\input{tables/reg_\baseapproach\basedir_\startyyyymm_\lastyyyymm_lnrvofpw_lnlampw_on_ead_lnwtw_bemew_momretw.tex}
\caption{\label{tbl:lnrvofpw_lnlampw_on_ead_lnwtw_bemew_momretw}Order flow volatility is higher and price impact is lower on earnings announcement days}
\end{subtable}
}
\par 
{\footnotesize Notes: Panel (a) shows regressions of the log value of
information and its components on size, book-to-market, and an earnings
indicator. \loginfovalregs \earnings%
\neglambda \stderror \stars 
}
\end{table}

Table \ref{tbl:lnrvofpw_lnlampw_on_ead_lnwtw_bemew_momretw} sheds further light
on the increase in the value of information on earning announcement days.
Here, we regress log order flow variance and log price impact on the earnings-day indicator
and control variables. In all specifications, we find a significant increase
in order flow volatility and a decrease in price impact on earning announcement days.
These two effects compete with each other, but because the increase in order flow volatility is larger, the value of information increases.

That the value of information rises on earnings announcement days is interesting in its own right. But it can also serve to explain the previously documented fact that investor attention spikes around these events. \textcite{ben-rephael_it_2017} and \textcite{liu_retail_2019} find that both institutional and retail investor attention rises on earnings days, and that this effect is stronger for institutional investors. This fact is consistent with institutional investors superior ability to gather and exploit private information. 

That price impact falls on earnings days is a novel and somewhat surprising finding. A reasonable prior would be that fundamental volatility is greater on earnings days and that market makers would respond by increasing price impact to protect themselves against adverse selection. But a counterforce is the intensity of noise trading on these days which appears to rise simultaneously.
The decline in price impact is also consistent with uninformed noise traders overestimating the novelty of earnings news \parencite{treynor_only_1971}. We are not aware of any other work that has documented this effect and believe it warrants further investigation elsewhere.

\subsection{\label{sub:infoval_iti}Informed trading intensity}

We next ask whether when our measure of the value of information is elevated, 
informed traders actually trade more. 
Building on \textcite{collin-dufrense_prices_2015}, 
\textcite{bogousslavsky_informed_2024} train a
machine learning model on trades by investors known ex post to be
informed---Schedule 13D filers, corporate insiders, and short sellers---and
use it to estimate a daily informed trading intensity (ITI) for each stock.
Table \ref{tbl:ITI_on_lndPdYw} regresses each of their five ITI measures on
the log value of information with date and stock fixed effects. All five
coefficients are positive and statistically significant: a one standard
deviation increase in the log value of information is associated with a 0.2
to 0.6 standard deviation increase in informed trading intensity within stock
and day. The relation is strongest for the impatient and short-seller
variants, consistent with our measure capturing the value of short-lived
private information.

\begin{table}
\caption{Informed trading intensity is high when information is valuable}
\label{tbl:ITI_on_lndPdYw}
\par
{\footnotesize
\input{tables/validation_ITI_on_lndPdYw_\baseapproach\basedir_200309_201907.tex}
}
\par
{\footnotesize Notes: Each column regresses an informed trading intensity
(ITI) measure of \textcite{bogousslavsky_informed_2024} on the log value of
information. \infoval ITI is a machine-learning estimate of informed trading
trained on trades by Schedule 13D filers (column 1), its impatient and
patient variants (columns 2--3), corporate insiders (column 4), and short
sellers (column 5). All variables are standardized, so coefficients are in
standard deviation units. The sample is September 2003 to July 2019, when
the ITI data end. \stderror \stars
}
\end{table}

\subsection{\label{sub:attention}Search attention}

A complementary external benchmark is investor attention. \textcite{da_market_2025}
construct a daily aggregate retail attention (ARA) measure from abnormal Google
search activity for individual stocks, and show that it negatively predicts
one-week-ahead market returns through retail-driven price pressure. ARA is
market-wide, built from search behavior, and entirely independent of the
microstructure inputs to our measure, which makes it a natural out-of-sample
check: if our statistic captures information-rich environments, it should be
elevated precisely when investors are paying attention. This also extends the
earnings-day evidence above, where the value of information rose exactly when
attention spikes \parencite{ben-rephael_it_2017,liu_retail_2019}, from scheduled
announcements to the daily market-wide level.

We aggregate our stock-day value of information into a daily equal-weighted
average and relate it to ARA over July 2004 to December 2019, the span of the ARA
series. Table \ref{tbl:aggval_on_ARA} reports the results: Panel A regresses the
log aggregate value of information on ARA in levels, while Panel B regresses the
daily change in each on the change in ARA, which removes any common trend. The two
series comove strongly and positively. A one standard deviation increase in retail
attention is associated with a 0.14 higher log value of information (about 15\%) in
levels, and a 0.06 (about 6\%) higher value in same-day changes. Both relations are
highly significant under Newey--West standard errors, and---unlike a spurious
trend---the first-difference relation survives. The coefficients are
stable across subperiods, including when we exclude the 2008--2009 financial
crisis, so the comovement is not an artifact of a single turbulent episode.

\begin{table}
\caption{The value of information is high when retail attention is high}
\label{tbl:aggval_on_ARA}
\par
{\footnotesize
\input{tables/validation_aggval_on_ARA_subsamples_\baseapproach\basedir_20040701_20191231.tex}
}
\par
{\footnotesize Notes: This table relates the daily equal-weighted average across
stocks of the log value of information to the aggregate retail attention (ARA)
measure of \textcite{da_market_2025}. \infoval Panel A reports levels regressions
of the aggregate log value of information on ARA; Panel B reports first-difference
regressions of its daily change on the change in ARA. ARA is standardized to unit
standard deviation over the sample, so coefficients are per one standard deviation
of attention. Each column reports a different sample window; the final column
excludes the 2008--2009 financial crisis. The sample is July 2004 to December
2019, the span of the ARA series. Newey--West heteroskedasticity- and
autocorrelation-consistent standard errors are in parentheses. \stars
}
\end{table}

Overall, the previous results validate that our proposed measure of the value of information is a useful and informative statistic, and that it yields results that are consistent with prior work.

\section{\label{sec:puzzle}The value of information puzzle}

\subsection{Economic magnitude of the value of information\label{sub:magnitude}}

We next step back to ask if the magnitude of the value of information that we estimate is consistent with that of the demand for information. Recall that from Table \ref{tbl:summarystats}, the average value of information per stock is about \$\EVwannualUSDm~million per year. This seems quite small, given that the average market capitalization of a stock in our sample is about \$\meanmktcapUSDb~billion. But, perhaps the value is considerably larger for some stocks. 

\begin{figure}
  \caption{Value of information by stock}
  \label{fig:infoval_bypermno}
  \begin{center}
  \includegraphics[width=1\columnwidth]{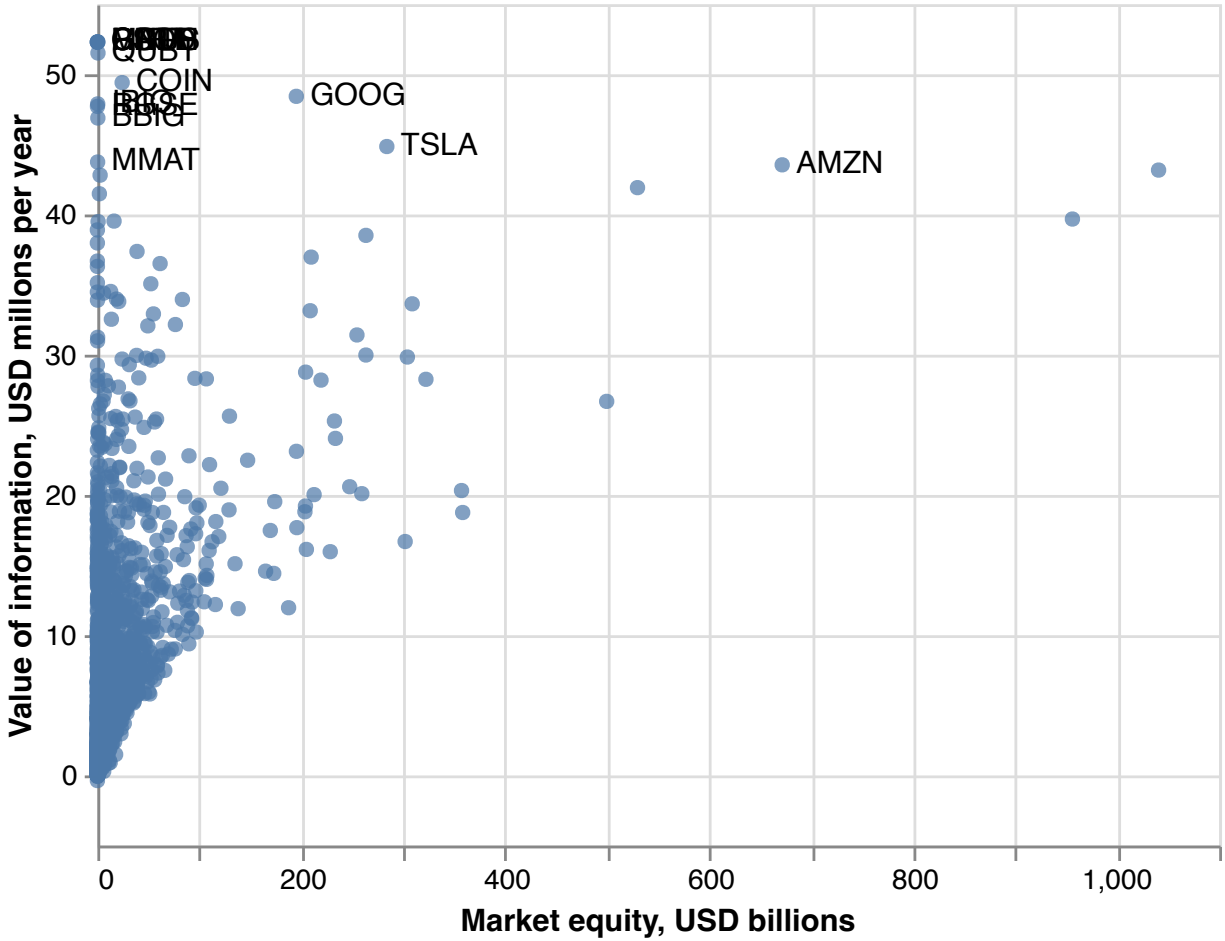}
  \end{center}
  \par
  {\footnotesize Notes: Mean value of information plotted against mean market equity 
  for each specific stock (permno). \infoval We label the 15 most information valuable
  stocks with their latest ticker symbols. 
  }
\end{figure}
Figure \ref{fig:infoval_bypermno} shows that indeed for some stocks, the value
of information is considerably larger. For example, the value for information on Tesla (TSLA) is about an order of magnitude larger than the average. But even for the most valuable stocks, the value of information is only a small fraction of market capitalization. To see this, in Figure \ref{fig:infovalpctofmarketcap} we plot the aggregate value of information as a percentage of aggregate market capitalization.

\begin{figure}
\caption{Value of information as percentage of market capitalization}
\label{fig:infovalpctofmarketcap}
\begin{center}
\includegraphics[width=1\columnwidth]{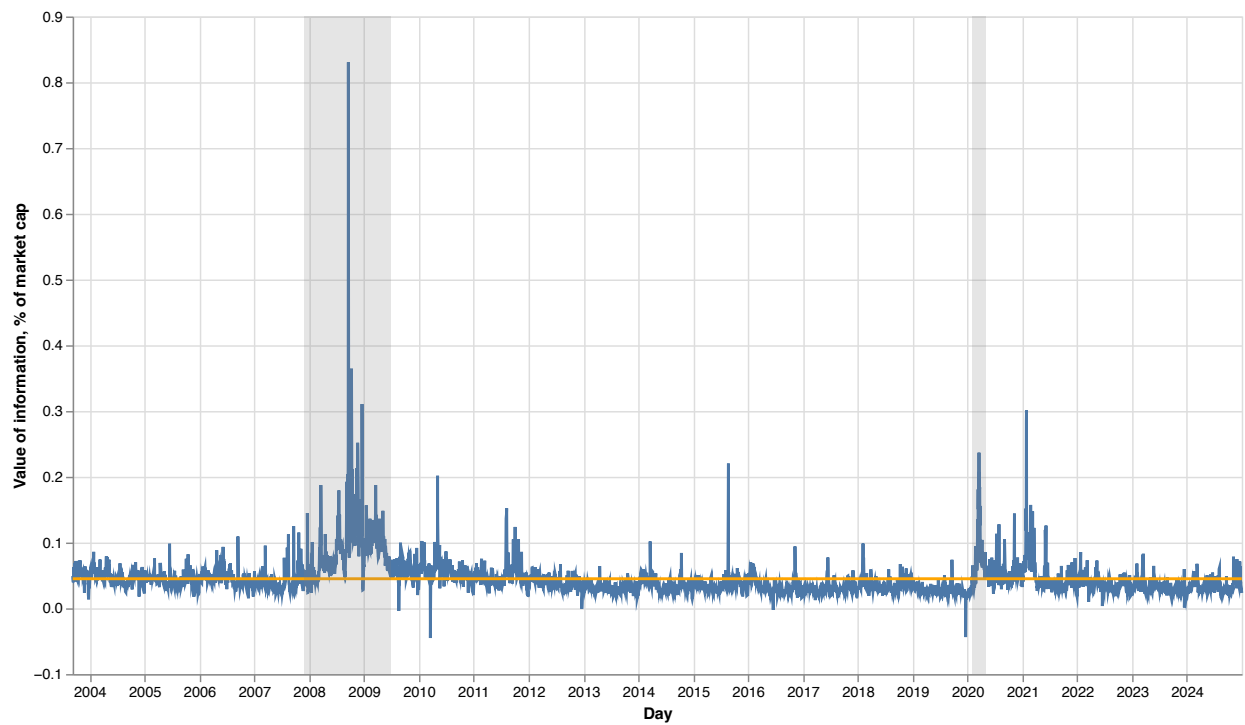}
\end{center}
\par  
{\footnotesize Notes: Plotted is the mean monthly information value, aggregated over stocks and days, and divided by the aggregate market capitalization in the prior month. The solid horizontal line is its full sample mean (\meaninfovalpctofmarketcap\%). \shades
}
\end{figure}

We estimate that the mean aggregate value of information is \meaninfovalpctofmarketcap\% of aggregate market capitalization. 
This aggregate statistic is equivalent to the market cap weighted ratio of the value of information to market capitalization.
From Figure \ref{fig:infovalpctofmarketcap}, we see that the value as percentage of market capitalization is remarkably stable over our sample period, 
though it tends to spike during turbulent times. Such spikes are partly mechanical, as a drop in stock market prices reduces the denominator in the ratio of the value of information to market capitalization.

Our estimates imply that the most active investors and asset managers can hope to gain by seeking superior information to forecast stock prices is only \meaninfovalpctofmarketcap\% on average of market capitalization. This estimate is remarkably small compared to the fees that active asset managers charge.
For example, \textcite{french_presidential_2008} estimates that investors pay active asset managers about \meanincrcostactivepctfrenchsample\% of market capitalization in fees each year searching for superior returns. We find it puzzling therefore that investors would pay so much more in fees than active asset management can possibly generate.

To further see why this is puzzling, consider the model by \textcite{garleanu_efficiently_2018}, where informed traders, uninformed noise traders, and market makers are joined by asset managers. They show that in a general equilibrium for both assets and asset management, the value of information equals management fees plus the cost of searching for good managers:
\begin{equation}
\text{Value of information} =  \text{Management fees} + \text{Search costs} \ge \text{Management fees}. \label{eq:voi_geq_fees}
\end{equation}
Assuming search costs are nonnegative, this implies that the value of information is at least as large as the fees charged by asset managers.
But our estimate of the value of information is much smaller than the fees charged by asset managers.

Note that our puzzle is different from the already puzzling finding in \textcite{french_presidential_2008} that the fees paid to active asset managers do not justify their average outperformance of passive investment strategies. Our findings imply that these fees are higher than what active asset managers could \emph{possibly} generate.

In Appendix \ref{sec:robustness} we show that this puzzling fact is robust to a battery of robustness tests. First, we study the sensitivity of our estimates to how we sign trades as buys or sells. In \autoref{fig:robust_Vw} we report the mean value of information across days and stocks for \signalgos. None of these have a major effect on the mean value of information. Second, in our baseline analysis we omit penny stocks with less than \$5 in price. The figure further shows that including such stocks would reduce the mean value of information. This is intuitive as penny stocks have little order flow and therefore little information value.
In \autoref{fig:robust_Vw_frequency} we show that our estimates are also robust to variation in the frequency of observation.

\subsection{Potential resolutions}

What can possibly explain the fact that the value of information that we estimate is so small compared to the fees that active asset managers charge? We next consider several potential resolutions.

\subsubsection{\label{subsub:beyond_French}Decline in fees}
One potential concern is that \textcite{french_presidential_2008}'s estimates of aggregate fees as a percentage of market capitalization are outdated. If active asset management fees have declined considerably since 2006, when French's sample ends, this could resolve the puzzle because we measure the value of information from \sampleperiod.

To address this concern, we replicate French's calculation and extend it through 2024. Appendix \ref{sec:french_replication} describes the data and methodology. Figure \ref{fig:french2008_cost_of_active} plots the resulting series. The incremental cost of active investing---the difference between what investors pay for their actual portfolios and what they would pay holding a passive market portfolio---averaged \meanincrcostactivepctfrenchsample\% of market capitalization over French's 1980--2006 sample and \meanincrcostactivepctextension\% over 2007--2024, as ETFs took share from higher-fee mutual funds, retail trading went commission-free, and indexing displaced institutional active management. The dollar cost of price discovery nevertheless kept growing, from \$7 billion in 1980 to \$135 billion in 2024, because the market grew faster than its per-dollar cost fell. 

The mean incremental cost of active investing over the \sampleperiod~ period which coincides with our value of information sample is \meanincrcostactivepctvoisample\% of market capitalization, which is still an order of magnitude larger than our estimate of the aggregate value of information.
Evidently, the fee compression observed since French's sample period does not resolve the puzzle.

\begin{figure}
\caption{The incremental cost of active investing, 1980--2024}
\label{fig:french2008_cost_of_active}
\begin{center}
\includegraphics[width=1\columnwidth]{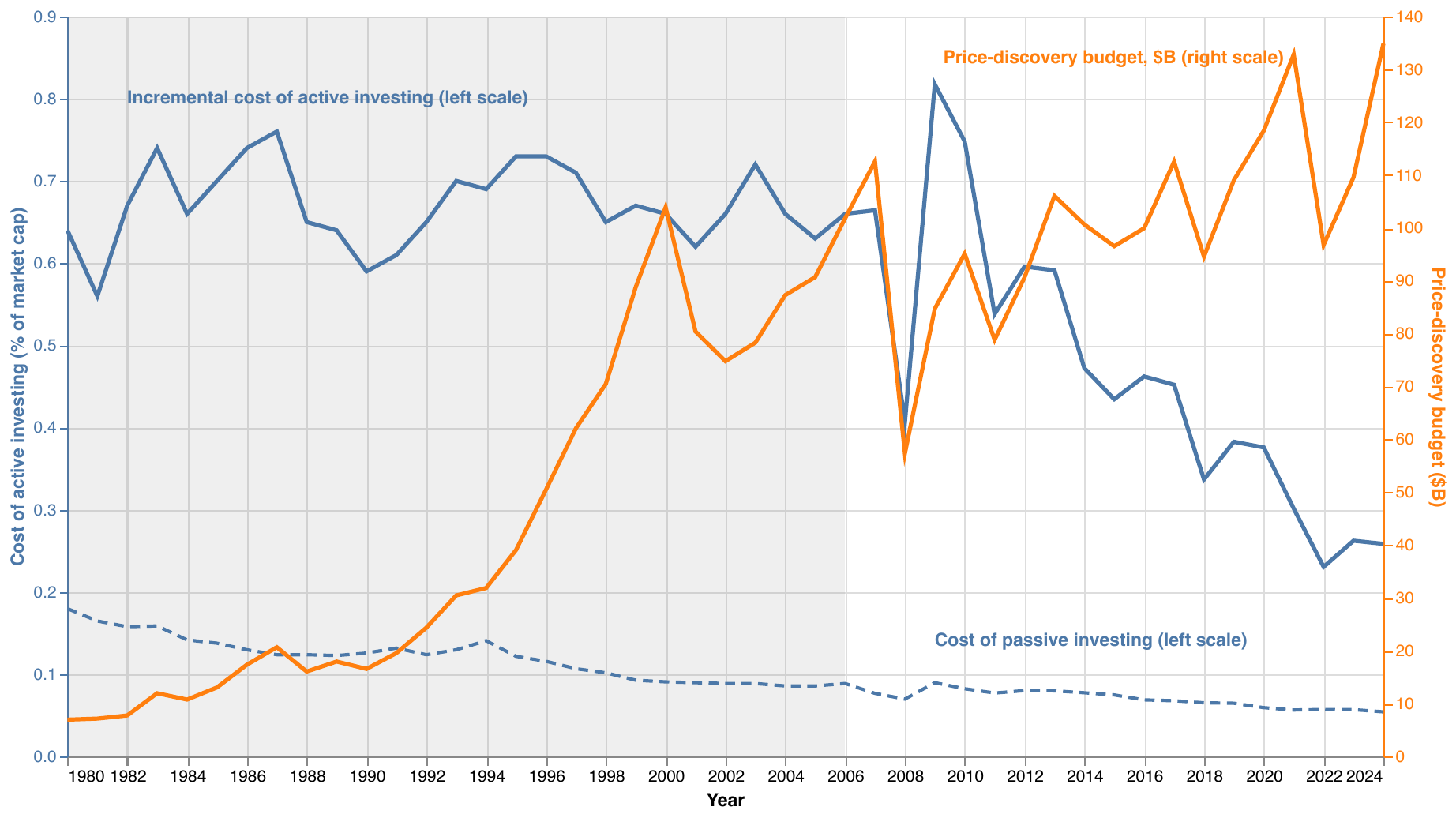}
\end{center}
\par
{\footnotesize Notes: The solid blue line is the incremental cost of active investing---the cost of actual investor portfolios minus the cost of a passive market portfolio---as a percentage of US equity market capitalization, replicating and extending \textcite{french_presidential_2008}. The dashed blue line is the cost of the passive portfolio. The orange line is the implied dollar cost of price discovery in billions (right scale). The shaded region marks French's original 1980--2006 sample, which our series reproduces exactly. Appendix \ref{sec:french_replication} describes the 2007--2024 extension.}
\end{figure}

\subsubsection{\label{subsub:pfof}Payment for order flow and price improvement}

In modern equity markets, retail brokers often receive payment for order flow (PFOF) from market makers who in addition provide traders on their platform with price improvement (a reduction in the bid-ask spread). \textcite{schwarz_actual_2025} estimate \$3.1 billion total PFOF in 2021 (equities + options) over five major brokers, with equities alone only ~\$0.5--0.7 billion since roughly two-thirds of PFOF comes from options. This is less than a basis point of stock market capitalization. They further estimate that price improvement is about 35\% of the NBBO bid-ask spread, and this figure would have to be further multiplied by the share of retail orders out of total order flow.\footnote{For example, \textcite{kadan_trading_2018} report that on the NYSE, less than 2\% of trade volume is attributed to retail traders.} %

PFOF effectively lowers price impact $dP$ for some traders, and lowers their cost of trading. The concern is that this could bias downward our measure of noise trader losses, $E \int_{0}^{T} dP_{t} dY_{t}$. Recall, however, that TAQ data include all trades, including those that receive price improvement. Thus, our measure of $dP$ is not affected by price improvement.

Finally, we note that PFOF and price improvement are not free. They are paid for by market makers like Citadel who provide them. They would not be willing to provide these payments and discounts unless they expected to at least break even on the trades they receive. How these payments and transfers occur is irrelevant. The important object to measure is the total effective losses of the side of the market that loses on average --- the uninformed noise traders, because that is all that is available for informed traders to extract. PFOF, therefore, does not resolve the puzzle.

\subsubsection{\label{subsub:other_benefits}Additional benefits of asset management}

One might argue that the fees investors pay to asset managers compensate not just for private information, but for a bundle of services: professional portfolio construction, diversification, factor or style exposures, liquidity transformation, tax management, custodial services, and convenience. If investors value these non-informational services, perhaps the gap between our estimate of the value of information and aggregate fees is not puzzling after all.

However, this argument does not resolve the puzzle. Passive index funds provide essentially the same bundle of services---diversification, factor exposure, liquidity, custody, and convenience---at dramatically lower cost. 
The incremental cost of active investing calculated by French and extended through 2024, which we report in Figure \ref{fig:french2008_cost_of_active}, represents the incremental cost investors pay specifically for active management---that is, for the promise of superior returns through private information or stock-picking skill over and above the cost of passive asset management.

Our estimates suggest that the aggregate value of information available in the market is only about \meaninfovalpctofmarketcap\% of market capitalization. This is far smaller than the active-passive cost differential, let alone the total fees charged by active managers. Thus, even if we grant that some portion of active management fees compensates for non-informational services, the portion attributable to active management remains puzzlingly large relative to the value such activity can generate.

\subsubsection{\label{subsub:some_days}Perhaps being active on some days is worth it?}

One might argue that the puzzle could be resolved if active managers strategically concentrate their information acquisition on days when the value of information is particularly high. Earnings announcement days are a natural candidate, as we documented in Section \ref{sub:infoval_earnings} that the value of information increases more than 3-fold on these days.

However, even this concentration strategy cannot bridge the gap. On earnings days, the value of information rises from \meaninfovalpctofmarketcap\% to roughly $\meaninfovalpctofmarketcap \times \experneffectonvaluebothfe \approx 0.14$\% of market capitalization. This is still considerably less than the \meanincrcostactivepctfrenchsample\% in fees that investors pay searching for superior returns.

Moreover, earnings days represent only a tiny fraction of trading days---about \meanernpct\% of stock-days in our sample. A selective strategy that captures information value only on these high-value days would therefore earn approximately $0.01 \times \experneffectonvaluebothfe \times \meaninfovalpctofmarketcap \approx 0.0014$\% of market capitalization per year---far less than the \meaninfovalpctofmarketcap\% one would earn by being active on all days. This calculation underscores that even on the most information-rich days in a firm's calendar, being active does not come close to justifying the fees charged by active asset managers.

\subsubsection{\label{subsub:informationarrival}Overnight, intraday, and dynamic information arrival}

In the Kyle-Back framework, the informed investor learns the fundamental value of the asset at the beginning of the trading day and trades on this information throughout the day. Mapping this to the real world, it is most natural to this on information arriving overnight, such as typical earnings announcements or macroeconomic releases. Our measures accurately captures the value of information in such cases using the losses of noise traders during the trading day because trading during market hours is the only way to profit from information that arrives overnight.

In reality, information may arrive at any point in time, and informed traders need to pay a cost to acquire the informative signal. Such situations are studied in \textcite{banerjee_strategic_2020} and \textcite{banerjee_dynamics_2022}. In their frameworks, informed trading is locally deterministic implying that our Assumption \ref{as:orthogonal_dX} is satisfied, and $E \int_{0}^{T} dP_{t} dY_{t}$ still measures the expected loss of noise traders. However, the interpertation of our Assumption \ref{as:competitive} is more nuanced. Our measure, $E \int_{0}^{T} dP_{t} dY_{t}$, captures the willingness to pay for information, rather than the net value after the cost of acquiring it. The net value is naturally smaller. As a result, our estimates provide an upper bound on the net gain of the informed traders in such models, making the puzzle even more pronounced.

More generally, our approach of estimating noise traders' losses, rather than attempting to estimate the value of information directly, makes the puzzle particularly robust. In Kyle–Back-type models, the losses of noise traders are ultimately the only source of profits for any market participant. Consequently, as long as our measure of noise-trader losses is valid, the puzzle remains because informed trader profits cannot exceed these losses.

\subsubsection{\label{subsub:oligopolistic_mm}Noncompetitive market makers}

Our first main assumption in deriving Proposition \ref{prop:voi_dPdY} is that the
total expected gain to informed investors equals the total expected losses to noise traders (Assumption \ref{as:competitive}). This equality holds if market makers are
competitive and earn no rents in equilibrium. It is therefore natural to consider relaxing this assumption. Doing so, however, would deepen the puzzle. 

Suppose that instead, market makers have some market power and charge a markup. In such a market, market makers receive the same amount of losses from noise traders, but pass only part of them to informed investors. As a result, aggregate losses by noise traders provide an upper bound on the value of information similar to the discussion in the previous sub-section.

\textcite{glode_asymmetric_2016} microfound such rents. In their model, trade occurs through chains of intermediaries who make take-it-or-leave-it offers under asymmetric information, and each intermediary captures part of the gains from trade as compensation for mitigating adverse selection. Their model also illustrates a second channel: when informational gaps are large, trade breaks down with positive probability, so part of the cost that asymmetric information imposes on the uninformed is a deadweight loss rather than a transfer to the informed. Both channels imply that the losses of the uninformed can overstate the gains of the informed, reinforcing the conclusion that our estimate bounds the value of information from above.

\subsubsection{\label{subsub:oligopolistic_am}Noncompetitive asset managers}

The aggregate fees estimated by French simply measure aggregate fees paid by investors, and do not assume anything about the competitiveness of asset management. The basic intuition, formalized in Equation \eqref{eq:voi_geq_fees}, suggests that the value of information is bounded from below by the fees paid by investors to the asset managers in equilibrium, regardless of the competitiveness of the asset management. This intuition is clearly violated by our estimates.

\subsubsection{\label{subsub:limited_arbitrage}Limited arbitrage, inventory costs, and risk averse market makers}

In the Kyle-Back framework, market makers are unconstrained in their ability to arbitrage the asset. In reality, market makers are constrained by inventory risk and search costs.

For example, \textcite{kondor_liquidity_2019} study the dynamic interactions between liquidity providers' capital, the liquidity that these agents provide to other participants, and assets' risk premia. They show that in a framework with minimal frictions, in particular, no asymmetric information or leverage constraints, risk averse market makers (arbitrageurs) would set positive price impact coefficients ($\lambda$) that decrease with their wealth. That is, even if no asymmetric information is present, price changes could still covary with order flow, and we could misattribute this to the value of information.

We are aware of no theoretical work that combines their model with asymmetric information, but it is plausible that in such an extension of their model, the aggregate losses of noise traders must still account for the most that market makers can transfer to informed investors, after being compensated for these additional costs. In this case too, we would underestimate the value of information, and the puzzle would deepen.

\subsubsection{\label{subsub:partial_info}Partial information and multiple informed investors}

In the \textcite{kyle_continuous_1985} and \textcite{back_insider_1992} framework, a monopolistic informed investor learns the fundamental value of the asset perfectly, without noise. It is arguably more realistic that the informed receive only a noisy signal about the asset's fundamental value. Moreover, multiple investors can simultaneously acquire information about the asset. Would such features change the interpretation of our estimates?

\textcite{back_imperfect_2000} consider a more general model in which multiple investors receive partial and potentially correlated signals about an asset's fundamental value. Although this model is harder to solve and generates different dynamics of trade and prices, it generates the same expression for the aggregate losses of the uninformed noise traders (Equation \eqref{eq:EdPdZ}) and for the aggregate value of information to all informed investors combined \parencite[][Corollary 3]{back_imperfect_2000}. Therefore, our main result in Proposition \ref{prop:voi_dPdY} holds in this more general setting.

Intuitively, even with multiple informed investors and partial signals, the gains to the informed come entirely from the losses of the noise traders. Therefore, our approach of focusing on the losses of the noise traders allows us to abstract from these specific model features.

\subsubsection{\label{subsub:assumption2_violation}Discrete time and relaxing  Assumption \ref{as:orthogonal_dX}}

Assumption \ref{as:orthogonal_dX} holds exactly in continuous-time Kyle-Back models because informed trading is smooth relative to noise trading. But what if we consider a more general discrete-period setting where this assumption may be violated? Here we show that even in such cases, our measure provides an \emph{upper bound} on the value of information.

Consider a discrete-time analog where trading occurs over $N$ periods, similar to \textcite{kyle_continuous_1985} and \textcite{foster_strategic_1996}. Let $\Delta P_n = P_n - P_{n-1}$, $\Delta Y_n = Y_n - Y_{n-1}$, $\Delta X_n = X_n - X_{n-1}$, and $\Delta Z_n = Z_n - Z_{n-1}$ denote the changes in prices, order flow, informed trading, and noise trading in period $n$. Assume that price changes are linear in order flow:
\begin{equation}
\Delta P_n = \lambda_n \Delta Y_n, \label{eq:discrete_price}
\end{equation}
for some positive price impact coefficient $\lambda_n > 0$. We further assume that noise trading innovations are uncorrelated with informed trading decisions:
\begin{equation}
E[\Delta Z_n \Delta X_n] = 0. \label{eq:orthogonal_ZX}
\end{equation}
This assumption is natural: noise traders' liquidity needs are exogenous and independent of informed traders' strategies, which depend on private information about fundamentals.

Due to Assumption \ref{as:competitive}, the expected profit to informed traders over all periods can be written as:
\begin{align}
\Omega &= E\sum_{n=1}^{N} \Delta P_n \Delta Z_n \nonumber \\
&= E\sum_{n=1}^{N} \Delta P_n (\Delta Y_n - \Delta X_n) \nonumber \\
&= E\sum_{n=1}^{N} \Delta P_n \Delta Y_n - E\sum_{n=1}^{N} \Delta P_n \Delta X_n. \label{eq:discrete_decomp}
\end{align}
The first term is precisely what we measure empirically. The second term captures the extent to which Assumption \ref{as:orthogonal_dX} is violated in discrete time.

Using the linear pricing rule \eqref{eq:discrete_price}, the second term becomes:
\begin{align}
E\sum_{n=1}^{N} \Delta P_n \Delta X_n &= E\sum_{n=1}^{N} \lambda_n \Delta Y_n \Delta X_n \nonumber \\
&= E\sum_{n=1}^{N} \lambda_n (\Delta X_n + \Delta Z_n) \Delta X_n \nonumber \\
&= E\sum_{n=1}^{N} \lambda_n (\Delta X_n)^2, \label{eq:violation_term}
\end{align}
where the last equality follows from the orthogonality condition \eqref{eq:orthogonal_ZX}. Since $\lambda_n > 0$, this term is non-negative. Therefore:
\begin{equation}
  \Omega = E\sum_{n=1}^{N} \Delta P_n \Delta Y_n - E\sum_{n=1}^{N} \lambda_n (\Delta X_n)^2 \le E\sum_{n=1}^{N} \Delta P_n \Delta Y_n. \label{eq:upper_bound}
\end{equation}

This establishes that our empirical measure $E\sum \Delta P_n \Delta Y_n$ provides an upper bound on the true value of information even when Assumption \ref{as:orthogonal_dX} does not hold exactly. If anything, violating this assumption would make the puzzle even deeper: the true value of information would be even smaller than our already modest estimates suggest.

Moreover, the ``error term'' in \eqref{eq:violation_term} vanishes as the time step shrinks. In Kyle-style models, informed trading increments $\Delta X_n$ are of order $\Delta t$ (smooth), so $(\Delta X_n)^2$ is of order $(\Delta t)^2$, which becomes negligible relative to $\Delta P_n \Delta Y_n$ (of order $\Delta t$) as $\Delta t \to 0$. This provides additional intuition for why Assumption \ref{as:orthogonal_dX} holds exactly in continuous time and approximately in high-frequency discrete time.

\subsubsection{\label{subsub:jumps}Jumps in prices}

The continuous-time Kyle-Back framework assumes that prices follow diffusion processes without jumps. In reality, asset prices can exhibit discontinuous jumps, particularly around news events or market stress. Importantly, the discrete-time analysis developed in Section \ref{subsub:assumption2_violation} implicitly accommodates such jumps. There, we only require that price changes satisfy $\Delta P_n = \lambda_n \Delta Y_n$ for some positive $\lambda_n$. This relation places no restriction on the magnitude or continuity of price changes between periods. Prices may jump discontinuously from $P_{n-1}$ to $P_n$, and the upper bound in \eqref{eq:upper_bound} still applies. Therefore, the result---that $E\sum \Delta P_n \Delta Y_n$ provides an upper bound on the value of information---extends naturally to settings with price jumps.

This generalization is empirically relevant because our high-frequency estimates aggregate price changes and order flow over discrete time intervals (e.g., one-minute windows). Within each interval, prices may exhibit jumps due to news arrivals or large trades. The discrete-time framework ensures that our estimates remain valid upper bounds on the value of information even in the presence of such discontinuities.

\subsubsection{\label{sub:multiasset}Multiple assets with cross-asset information}%

A natural concern with our estimates is that they are built asset by asset. We
measure, for each stock $i$, the covariation between its price change and its own
order flow, and sum across stocks. Yet information is rarely about a single
asset in isolation. A signal about firm $i$ is typically informative about the fundamental value of
many firms $j\neq i$ at once. When an informed investor pays once for a signal and
trades a whole portfolio, does our stock-by-stock measure capture the full value
of that information, or does it miss the cross-asset component? We show that, under
the multi-asset analogs of Assumptions
\ref{as:competitive} and \ref{as:orthogonal_dX}, the aggregation of the stock-by-stock measure is exactly
right: summing each stock's own price-change/order-flow covariation recovers the
\emph{total} value of information, including all cross-asset spillovers. The
spillovers do not have to be measured separately because each stock's price
already responds to the entire vector of order flows.

Intuitively, what we measure empirically is the total losses to noise traders across all stocks. These losses must cover both types of gains to informed traders: (i) gains from trading in a stock $i$ based on a signal on stock $i$'s fundamental value; and (ii) gains from trading in stock $i$ based on a signal on stock $j\neq i$. As long as market makers break even in aggregate, summing over all $i$, the expected total losses to noise traders must equal the expected total value of information across all stocks. 

\paragraph*{Setup}

There are $N$ assets traded over $[0,T]$ with a random fundamental value vector
$v=(v_{1},\dots,v_{N})'$ drawn from a joint distribution $F$. We place no
restriction on $F$; in particular the components of $v$ may be arbitrarily
correlated, so that a signal about one coordinate is informative about the
others---this is the cross-asset informativeness we wish to capture. One or more
informed traders observe signals about $v$ at $t=0$ and choose a vector of
cumulative orders $X_{t}=(X_{1t},\dots,X_{Nt})'$, a continuous It\^o process.
Noise traders submit a cumulative order vector $Z_{t}$ with
\begin{equation}
dZ_{t}=\Sigma_{z}(t)^{1/2}\,dB_{t},\nonumber
\end{equation}
where $B_{t}$ is an $N$-dimensional standard Wiener process and $\Sigma_{z}(t)$
is the (possibly stochastic) instantaneous covariance matrix of noise trading.
Off-diagonal entries of $\Sigma_{z}$ allow correlated liquidity demand across
assets. Market makers observe the full order
flow vector $Y_{t}=X_{t}+Z_{t}$ and set the price vector
$P_{t}=(P_{1t},\dots,P_{Nt})'$, a continuous It\^o process. Crucially, the price
of asset $i$ may depend on the order flow of \emph{all} assets, since order flow
in asset $j$ can be informative about $v_{i}$.

\paragraph*{The value of information}

A noise trader buying $dZ_{it}$ shares of asset $i$ when its price moves by
$dP_{it}$ loses $dP_{it}\,dZ_{it}$; losses accrue only on the asset actually
traded. Summing over assets, the total expected loss of all noise traders is
$E\int_{0}^{T}\sum_{i}dP_{it}\,dZ_{it}=E\int_{0}^{T}dP_{t}'\,dZ_{t}$. Competitive
market making aggregated across the $N$ markets then gives the multi-asset analog
of Assumption \ref{as:competitive}.

\begin{assumption}[Competitive market making, multi-asset]\label{as:competitive_N}
The total value of information to all informed traders combined equals the total
expected loss to all noise traders,
\begin{equation}
\Omega=E\int_{0}^{T}dP_{t}'\,dZ_{t}
      =E\int_{0}^{T}\sum_{i=1}^{N}dP_{it}\,dZ_{it}.\nonumber
\end{equation}
\end{assumption}

As in the single-asset case, this is a wealth-conservation identity: fundamental
value is redistributed among informed traders, noise traders, and market makers,
and competitive (zero-profit) market making across all markets makes informed
gains equal noise losses in aggregate. This is a weak condition as it imposes zero market-making profits in \emph{aggregate} rather than market by market.

The orthogonality assumption generalizes in a straightforward manner. As long as each informed
holding $X_{it}$ is smooth---of order $dt$ rather than $\sqrt{dt}$---its quadratic
covariation vanishes against \emph{any} It\^o price process, own or cross.

\begin{assumption}[Orthogonal informed trading, multi-asset]\label{as:orthogonal_N}
Informed trading is orthogonal to prices in the sense that the matrix quadratic
covariation of $X$ and $P$ is zero,
\begin{equation}
\int_{0}^{T}dX_{it}\,dP_{jt}=0\qquad\text{for all }i,j.\nonumber
\end{equation}
\end{assumption}

Assumption \ref{as:orthogonal_N} thus holds for the same reason as Assumption \ref{as:orthogonal_dX}: smooth informed trading does not contribute to the diffusive part of any price. Note that only the diagonal
entries $[X_{i},P_{i}]_{T}=0$ are needed for the result below; the off-diagonal
entries vanish too, which we use only for interpretation.

\begin{proposition}[Value of information, multi-asset]\label{prop:voi_dPdY_N}
Under Assumptions \ref{as:competitive_N} and \ref{as:orthogonal_N}, the total
value of information over an interval of length $T$ is
\begin{equation}
\Omega=E\int_{0}^{T}dP_{t}'\,dY_{t}
      =\sum_{i=1}^{N}E\int_{0}^{T}dP_{it}\,dY_{it}.\nonumber
\end{equation}
\end{proposition}

\begin{proof}
Write $dZ_{t}=dY_{t}-dX_{t}$. Then
\begin{equation}
\int_{0}^{T}dP_{t}'\,dZ_{t}
=\int_{0}^{T}dP_{t}'\,dY_{t}-\int_{0}^{T}dP_{t}'\,dX_{t}
=\int_{0}^{T}dP_{t}'\,dY_{t},\nonumber
\end{equation}
where the last equality uses the diagonal of Assumption \ref{as:orthogonal_N}.
Taking expectations and invoking Assumption \ref{as:competitive_N} gives the
result.
\end{proof}

Proposition \ref{prop:voi_dPdY_N} is the multi-asset counterpart of Proposition
\ref{prop:voi_dPdY}, and it has a sharp empirical implication. The right-hand side
is exactly our estimator summed across stocks: for each stock $i$ we compute the
covariation between its \emph{own} price change and its \emph{own} order flow, and
we add up over stocks. No cross-asset covariation $\int dP_{it}\,dY_{jt}$ with
$i\neq j$ ever appears. The cross-asset information channel is nonetheless fully
counted, because $dP_{it}$ already embeds the response of asset $i$'s price to the
order flow of all assets.

\subsubsection{\label{subsub:risk_aversion}Risk averse informed investors}

In Kyle-Back models, informed investors are risk neutral. By contrast, informed investors are risk averse in \textcite{subrahmanyam_risk_1991}'s extension of the Kyle model, and in noisy rational expectations models in the tradition of \textcite{grossman_impossibility_1980}. How much would risk averse informed investors be willing to pay in this environment? Well, the aggregate losses of the noise traders would still measure the expected gain in monetary terms from superior information. But due to risk aversion and Jensen's inequality, the informed traders would be willing to pay less than their expected gain. Thus, higher risk aversion would lower the value of information and further deepen the puzzle. 

\subsubsection{\label{subsub:sdf}Correlation with the marginal value of wealth}

The preceding argument shows that risk aversion reduces the certainty-equivalent value of risky trading profits below their physical-measure expectation. But risk aversion has a second channel that can work in the opposite direction. If informed trading profits covary positively with the marginal value of wealth---that is, if profits are higher in states when consumption is scarce---then the risk-adjusted value of these profits can exceed their physical-measure expectation.\footnote{%
Valuing information with state prices relates to the probability-pricing framework of \textcite{davila_probability_2025}, who express the willingness-to-pay for marginal changes in probabilities---including changes in public and private information---as an asset-pricing formula with hypothetical cash flows.}

To formalize this, let $M_t>0$ denote a stochastic discount factor (SDF) assumed to be a continuous It\^o process such that $E(M_t)\le1$ for all $t$ (i.e., the risk-free rate is assumed non-negative). Denote by $\Omega^{M}$ the risk-adjusted value of information, which by Assumption \ref{as:competitive} equals the expected risk-adjusted losses to noise traders
\begin{equation} \label{eq:omegaM}
  \Omega^{M} = E\int_0^T M_t \, dP_t \, dZ_t.
\end{equation}

Note that what we measure empirically is $\Omega$, the undiscounted risk-neutral value of information. When $M_t$ is identically 1, $\Omega$ and $\Omega^{M}$ coincide. More generally, the next proposition establishes a relation between $\Omega$ and $\Omega^{M}$.

\begin{proposition}[Value of information with a stochastic discount factor]\label{prop:voi_SDF}
  Under Assumptions \ref{as:competitive}--\ref{as:linear_dP},  the value of information in the presence of a stochastic discount factor $M_{t}$
over an interval of length $T$ satisfies
\begin{align}
  \Omega^{M} &\le  \Omega + \int_0^T \text{Cov}(M_t, \, dP_t \, dY_t). \label{eq:sdf_decompose}
\end{align}
\end{proposition}

\noindent \textit{Proof:} See Appendix.

The first term on the right-hand-side of (\ref{eq:sdf_decompose}) is the risk-neutral value of information, which we measure empirically. The second term captures whether noise trader losses tend to be larger in high marginal utility states. When this term is zero (e.g., when $dP_tdY_t$ is non-stochastic), we get $\Omega^M\le\Omega$, which is essentially the result of the previous section, noting that risk-aversion is reflected here in the condition $E(M_t)\leq1$, which implies a positive risk free rate. By contrast, when the second term is positive, the risk-adjusted value of information is higher than the risk-neutral value of information. Figure~\ref{fig:infovalpctofmarketcap} is suggestive of this channel: aggregate noise trader losses are elevated during the 2008 financial crisis and the COVID-19 pandemic, periods in which the marginal value of wealth is plausibly high. The question is whether the second term, reflecting a risk adjustment, can be sufficiently positive to resolve the value-of-information puzzle.

To answer this question, we now bound how large this risk adjustment can be using a standard asset pricing calibration exercise. Applying the Cauchy-Schwarz inequality to (\ref{eq:sdf_decompose}), we have

\begin{align}
  \Omega^{M} &\le  \Omega + \int_0^T \sigma(M_t) \, \sigma(dP_t \, dY_t)dt. \label{eq:sdf_decompose1}
\end{align}

The asset pricing literature provides estimates of the entropy of the SDF in leading representative-agent models. Under the common simplifying assumption that the SDF is log-normal, $\log M_t \sim \mathcal{N}(\mu_t,\sigma^2_{\log M, t})$, the coefficient of variation of $M_t$ is $\sigma(M_t)/E[M_t] = \sqrt{e^{\sigma^2_{\log M, t}}-1}$. The corresponding entropy of the SDF in the sense of \textcite{alvarez_using_2005} is $L(M_t) := \log E[M_t] - E[\log M_t] = \sigma^2_{\log M, t}/2$, so that $\sigma(M_t)/E[M_t] = \sqrt{e^{2 L(M_t)}-1}$. Substituting in (\ref{eq:sdf_decompose1}) and using $E[M_t]\le 1$ yields the bound,
\begin{align}
  \Omega^{M} &\le \Omega + \sigma(\Omega)\,\sqrt{e^{2 L(M)}-1}, \label{eq:sdf_decompose2}
\end{align}
where $\sigma(\Omega) := \sigma\bigl(dP_t\,dY_t\bigr)$ denotes the standard deviation of cumulative informed-trading profits over the horizon $T$. We focus on estimates of $L(M)$ that are admissible by data on equity premia and bond yields. \textcite{backus_sources_2014} report annualized entropy in two leading specifications jointly consistent with the U.S. equity premium and the small mean yield spread. The first is the \textcite{campbell_by_1999} habit model (their Table III, column 3), which has $L(M)\approx 0.28$ per year, equivalently $\sqrt{e^{2\cdot 0.28}-1}\approx 0.86$ annually. The second is an iid-jump specification with recursive utility (their Table IV, column 1), which has the highest entropy among admissible models, with $L(M)\approx 0.58$ per year, equivalently $\sqrt{e^{2\cdot 0.58}-1}\approx 1.48$ annually.

We estimate the standard deviation of the value of information as a percentage of market capitalization, on an annual basis, at $\stdinfovalpctofmarketcapw$ based on the winsorized series and at $\stdinfovalpctofmarketcap$ based on the unwinsorized series. Plugging into (\ref{eq:sdf_decompose2}) with the largest admissible value of entropy, the maximum risk adjustment is bounded above by $\stdinfovalpctofmarketcap \times 1.48 \approx 0.04$ percentage of market capitalization. Combining with our estimate of $\Omega \approx \meaninfovalpctofmarketcap\%$, the implied upper bound on the risk-adjusted value of information is approximately $0.08\%$ of market capitalization---of the same order as $\Omega$ itself, and an order of magnitude smaller than the $\meanincrcostactivepctfrenchsample\%$ estimated by \textcite{french_presidential_2008}. The SDF covariance channel therefore points in the right direction but cannot quantitatively resolve the puzzle.

\subsubsection{\label{subsub:irrational}Behavioral explanations}

Our discussion above suggests that solving the puzzle requires an additional source of revenue for market makers that would compensate them for losses to informed traders, even if price impact is as low as our estimates suggest.
\textcite{treynor_only_1971} writes about a third type of traders that is potentially missing from the Kyle-Back framework---"transactors acting on information which they believe has not yet been fully discounted in the market price but which in fact has." \textcite{treynor_only_1971} explains,
\begin{quote}
    This is where the third kind of transactor comes in: From the market maker's point of view, his effect is identical to the liquidity-motivated transactor's. The market maker naturally welcomes the cooperation of wire houses and information services like the Wall Street Journal that broadcast information already fully discounted since many investors are easily persuaded to transact based on that information, hence enable the market maker to maintain substantially smaller spreads than would be possible without their trading activity.
\end{quote}

While \textcite{treynor_only_1971} thought of such traders as identical in their effect on the market to noise traders, one may imagine their misunderstanding of their position in the information supply chain,  may lead them to trade in the wrong direction on average. In this case, such behavioral investors would lose more than the bid-ask spread in each trade. 

Large systematic trading mistakes are a potential explanation for our puzzling estimates. In such an environment, the aggregate losses that are shared by market makers and informed traders are composed of both trading costs which we measure and trading mistakes which we do not. 

\textcite{barber_just_2009} quantify the losses of individual investors using complete transaction data from Taiwan. They find that individual investors lose approximately 2.3\% of aggregate market capitalization annually between 1995 and 1999. Of this, gross trading losses---the pure stock-picking losses from buying stocks that underperform stocks sold---amount to 0.6\% of market capitalization annually, with the remainder attributable to commissions, transaction taxes, and market-timing losses. \textcite{andries_financial_2025} study affluent French brokerage clients and find that unadvised clients generate a $-4.7\%$ annual risk-adjusted portfolio return. 

Evidence from US markets provides mixed evidence of retail investor performance. \textcite{kaniel_individual_2008} document positive excess returns in the month following intense buying by individuals and negative excess returns after individuals sell. \textcite{barber_attention-induced_2022} find that intense buying by Robinhood users forecasts negative returns, with average 20-day abnormal returns of $-4.7\%$ for the top stocks purchased each day. \textcite{barber_resolving_2024} show that retail purchases are concentrated in attention-grabbing stocks that subsequently underperform, with long-short strategies based on retail order imbalance yielding $-14.8\%$ annualized returns among stocks with heavy retail trading. While these studies document systematic trading mistakes by retail investors, they do not provide aggregate estimates scaled by market capitalization. Whether such losses are large enough to explain our puzzling estimates remains an open question.

\section{\label{sec:Conclusion}Conclusion}

We develop a simple, general closed-form expression for the value of information in financial markets. Under mild assumptions satisfied by a broad class of Kyle-Back models---competitive market making and orthogonal informed trading---the value of information equals the expected quadratic covariation between price changes and order flow. This expression is easy to estimate using high-frequency data and has straightforward economic units.

Using intraday TAQ data on US\ equities over \sampleperiod, we estimate that the average annualized value of information is approximately \$\EVwannualUSDm~million per stock. In aggregate, the value of information amounts to only \meaninfovalpctofmarketcap\% of market capitalization. Our estimates pass several validation tests: the value of information rises during turbulent periods, is higher for large, growth, and momentum stocks, and increases more than 3-fold on earnings announcement days. Moreover, the value of information is higher when retail attention and informed trading are higher.

Our findings present a puzzle. The aggregate value of information we estimate is considerably smaller than the approximately \meanincrcostactivepctfrenchsample\% of market capitalization that investors pay in fees to active asset managers each year \parencite{french_presidential_2008}. Even accounting for fee compression over time and focusing on the active-passive fee differential, this gap remains substantial. We consider several potential resolutions---including non-competitive market makers, risk aversion, partial information, and behavioral trading mistakes---but find that most would either deepen the puzzle or leave it unresolved.

One interpretation of our findings is that the equity market is highly informationally efficient, leaving little value for informed traders to extract. An alternative interpretation is that a significant portion of active management fees compensates for services unrelated to information, or that behavioral investors systematically lose more than the trading costs we measure. 
The wedge we document between the value of information and the cost of seeking it guides the construction of future theoretical models.

\clearpage

\begin{singlespace}
\printbibliography 
\end{singlespace}

\clearpage{}

\appendix

\setcounter{section}{0}
\setcounter{table}{0}
\setcounter{figure}{0}

\renewcommand{\theequation}{A.\arabic{equation}} \renewcommand{%
\thetable}{A.\arabic{table}} \renewcommand{\thefigure}{A.\arabic{figure}}

\section{Appendix}

\subsection{\label{sec:proofs}Proof of Proposition \ref{prop:voi_SDF}}

By (\ref{eq:omegaM}), 

\begin{align}
  \Omega^{M} &= E\int_0^T M_t \, dP_t \, dZ_{t} \nonumber \\
  &= E\int_0^T M_t \, dP_t \, (dY_t-dX_t) \nonumber \\
  &=E\int_0^T M_t \, dP_t \, dY_t-E\int_0^T M_t \, dP_t \, dX_t \nonumber \\
  &=E\int_0^T M_t \, dP_t \, dY_t,\label{eq:MdPdY}
\end{align}
where the last equality follows by Assumption \ref{as:diffusive_trading_strategies}, which implies that $M_{t}dP_{t}dX_t$ is zero, because $X_t$ has no stochastic term.

By Assumption \ref{as:linear_dP},
\begin{equation}
dP_t \, dY_t=\lambda_{t} (dY_t)^2 = \lambda_{t} \sigma_y^2\,dt, \label{eq:dPdY1}
\end{equation}
which is locally deterministic. 

The integral in \eqref{eq:MdPdY} is therefore a pathwise Lebesgue-Stieltjes integral, and \eqref{eq:MdPdY} can be written as 
\begin{equation}
  \Omega^{M}  =E\int_0^T M_t\,\lambda_t\,\sigma_y^2\,dt.
\end{equation}

By Fubini's theorem for Lebesgue integrals,

\begin{equation}
  \Omega^{M}  =\int_0^T E[M_t\,\lambda_t\,\sigma_y^2]dt.
\end{equation}

Now, at each fixed~$t$, $M_t$ and $\lambda_t\sigma_y^2$ are random variables over the relevant probability space, so the identity $E[XY]=E[X]\,E[Y]+\text{Cov}(X,Y)$ applied pointwise yields 

\begin{align}
  \Omega^{M}&=\int_0^T E[M_t]\,E[\lambda_t\sigma_y^2]dt+\int_0^T\text{Cov}(M_t,\lambda_t\sigma_y^2)dt\nonumber \\ 
  &\le\int_0^T E[\lambda_t\sigma_y^2]dt+\int_0^T\text{Cov}(M_t,\lambda_t\sigma_y^2)dt \nonumber \\
  &=\int_0^T E[dP_{t}dY_{t}]+\int_0^T\text{Cov}(M_t,dP_{t}dY_{t}) \nonumber \\
  &=E\int_0^T dP_tdY_t+\int_0^T\text{Cov}(M_t,dP_{t}dY_{t}), \nonumber \\
  &=\Omega+\int_0^T\text{Cov}(M_t,dP_{t}dY_{t}), \nonumber
\end{align}
where the inequality follows because $E(M_t)\le1$ and $\lambda_t>0$ (by Assumption \ref{as:linear_dP}), and the penultimate equality follows by applying Fubini's theorem to the left term. \hfill $\square$

\clearpage

\subsection{\label{sec:robustness}Robustness}

Here we include some additional robustness tests.

To study the sensitivity of our estimates to how we estimate price impact, in \autoref{fig:robust_Vw} we report the mean value of information across days and stocks for \signalgos. None of these have a major effect on the mean value of information.

Our baseline analysis omits penny stocks with less than \$5 in price. The figure further shows that including such stocks would reduce the mean value of information. This is intuitive as penny stocks likely have large price impact coefficients.

In \autoref{fig:robust_Vw_frequency} we show that our estimates of the value of information rise only slightly with the sampling interval between observations. As we show in Section \ref{subsub:assumption2_violation}, 
the error term in \eqref{eq:violation_term} vanishes as the time step shrinks. Therefore, higher frequency estimates provide a better approximation to the value of information, while lower frequency estimates are biased upwards. The figure shows that even at 30-minute frequency we estimate a modest upper bound on the value of information.

\autoref{fig:robust_regcoef} reports similar robustness tests, but this time for the regression analysis reported in the baseline regression tables in the paper (e.g. \autoref{tbl:reg_ead}). Again, we find that our conclusions are robust to these alternative specifications.

\begin{figure}[b!]
\caption{\label{fig:robust_Vw}Robustness: Mean value of information}
\begin{center}
\includegraphics[width=0.70\textwidth]{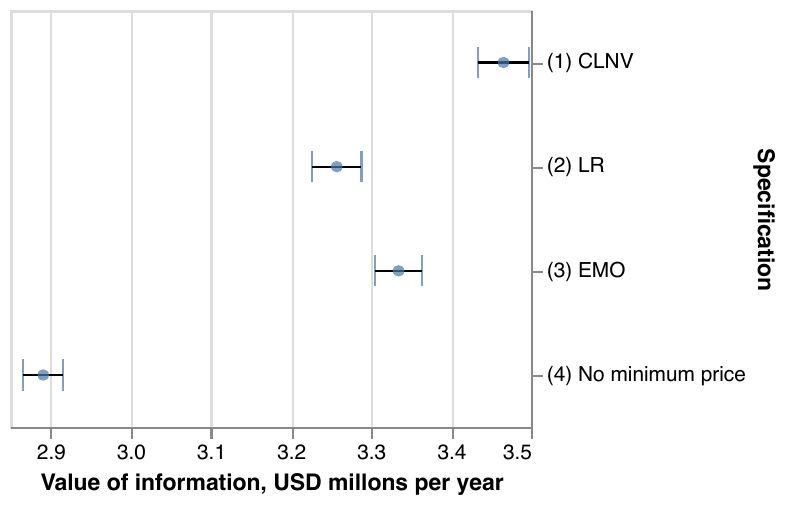}
\end{center}
\par
{\footnotesize Notes: Reported is the mean value of information across days and stocks
for \signalgos. \nominprice \idiovar \infoval \winsor 
}
\end{figure}

\begin{figure}[b!]
\caption{\label{fig:robust_Vw_frequency}Robustness to the frequency of observation}
\begin{center}
\includegraphics[width=0.70\textwidth]{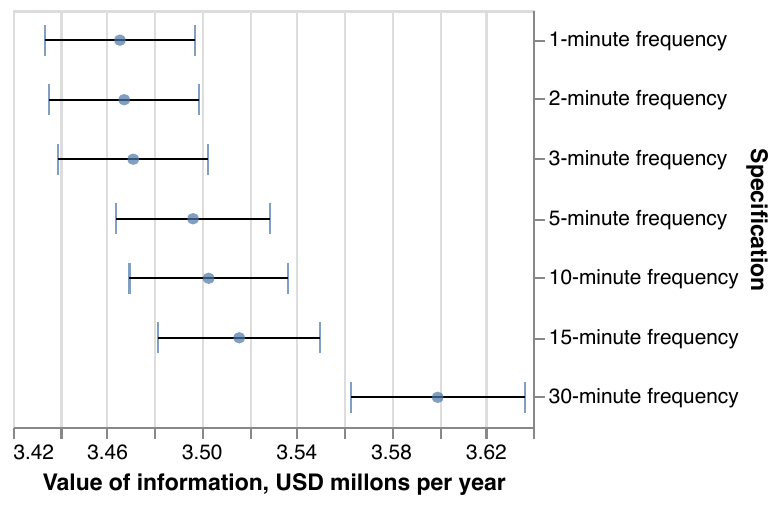}
\end{center}
\par
{\footnotesize Notes: Reported is the mean value of information across days and stocks
for various choices for the frequency of observation. \infoval \winsor 
}
\end{figure}

\begin{figure}
\caption{Robustness: Value of information and stock characteristics}
\label{fig:robust_regcoef}
\begin{center}
\begin{subfigure}[t]{0.49\columnwidth}
\includegraphics[width=1\textwidth]{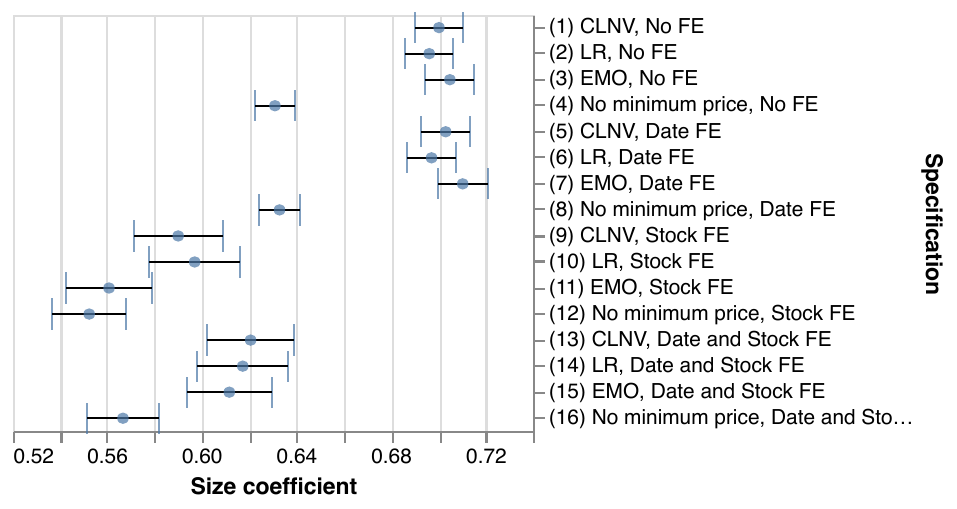}
\caption{\label{fig:robust_regcoef_lnwtw}Size}
\end{subfigure}
\begin{subfigure}[t]{0.49\columnwidth}
\includegraphics[width=1\textwidth]{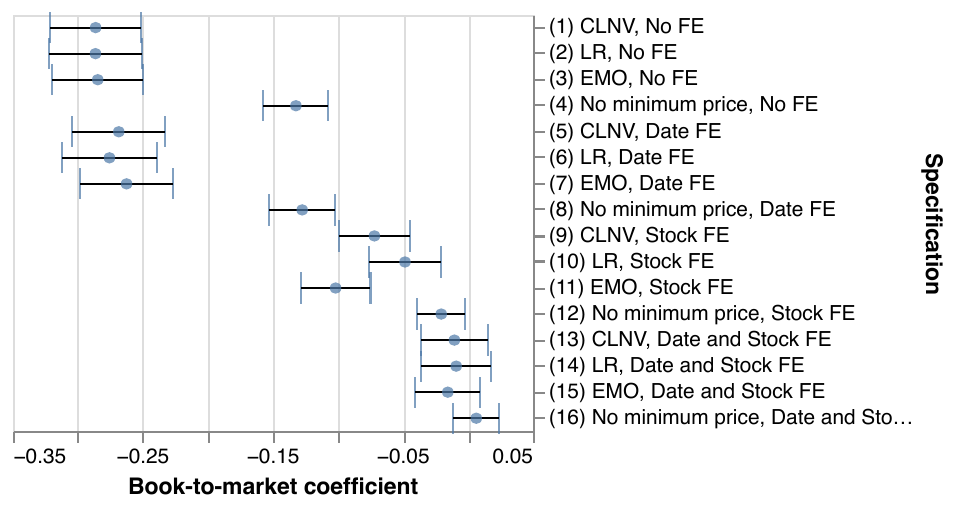}
\caption{\label{fig:robust_regcoef_bemew}Book-to-market}
\end{subfigure}
\par
\begin{subfigure}[t]{0.49\columnwidth}
\includegraphics[width=1\textwidth]{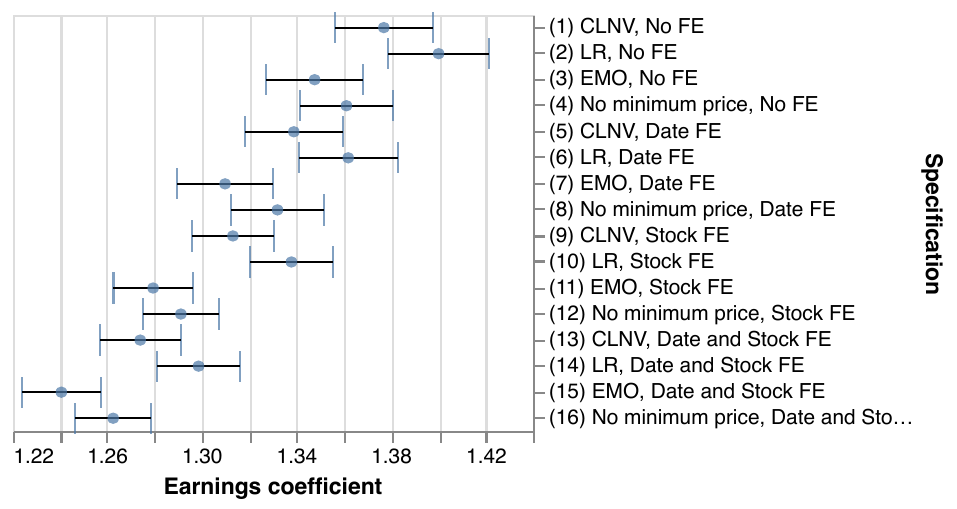}
\caption{\label{fig:robust_regcoef_ead}Earnings} 
\end{subfigure}
\begin{subfigure}[t]{0.49\columnwidth}
\includegraphics[width=1\textwidth]{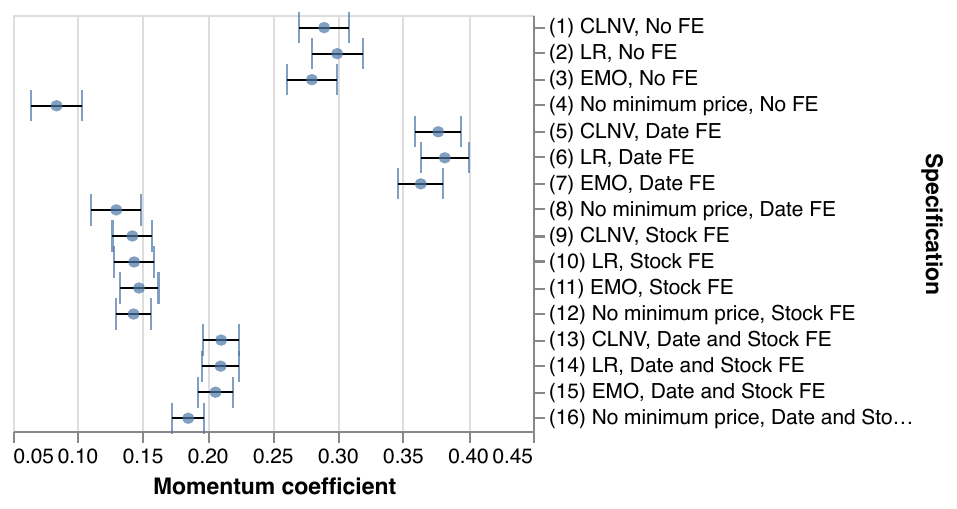}
\caption{\label{fig:robust_regcoef_momretw}Momentum}
\end{subfigure}
\end{center} 
\par
{\footnotesize Notes: Reported are coefficients from regressions of the log value of
information on stock characteristics. We evaluate several fixed effect specifications and \signalgos. \nominprice \size %
\momret \beme \neglambda \stderrorci
}
\end{figure}

\clearpage

\subsection{\label{sec:french_replication}Replication and extension of French (2008)}

\textcite{french_presidential_2008} measures the cost of active investing in US equities as the difference between what investors pay for their actual portfolios and what they would pay if all equity were held in a passive market portfolio. He combines the allocation of US common equity across investor types with the fees and expenses of each type and the US-equity trading revenue of broker-dealers, each scaled by aggregate market capitalization, over 1980--2006. We replicate his calculation and extend every component through 2024; Figure \ref{fig:french2008_cost_of_active} plots the resulting series.

For 1980--2006 we use French's published tables, so our series reproduces his results exactly and the incremental cost of active averages \meanincrcostactivepctfrenchsample\% of market capitalization. We verify the inputs we can observe independently: aggregate market capitalization recomputed from CRSP (NYSE, Amex, and Nasdaq common stocks, averaging the twelve beginning-of-month values each year) matches his within 0.6\% every year; the mutual fund allocation rebuilt from the current Z.1 Financial Accounts matches his Table I within 4.4 percentage points; and broker-dealer commissions and trading gains reconstructed from archived SEC FOCUS aggregates match his Table V within 15\% in 21 of 27 years.

For 2007--2024 we rebuild each component from public sources as explained below. Where French's source is proprietary, we anchor to his last published value and trend it on a closely related public series. Figures \ref{fig:french2008_active_cost_breakdown} and \ref{fig:french2008_passive_cost_breakdown} plot the resulting costs of active and passive investing by component; the incremental cost in Figure \ref{fig:french2008_cost_of_active} is their difference.

The mutual fund, closed-end fund, and ETF share of US corporate equity comes from the Z.1 Financial Accounts (Table L.224). The institutional share---private defined benefit and defined contribution plans, public pensions, and insurance companies---comes from Z.1 tables L.118 and L.224, plus hedge fund assets from HFR. The 2018 Z.1 redesign folded ESOPs, nonprofits, and bank personal trusts into the household sector, so our institutional share undercounts French's definition; restoring these sectors would raise the 2007--2024 incremental cost by at most 1.5 basis points.

Mutual fund fees are the value-weighted expense ratio plus annualized loads of US equity funds in the CRSP survivor-bias-free mutual fund database; the expense ratio falls from 78 basis points in 2007 to 29 in 2024, and annualized loads fall from 25 basis points to nearly zero. The institutional fee scales French's 2006 CEM Benchmarking value (23 basis points) by the ICI asset-weighted expense ratio of active equity funds, declining to 15 basis points by 2024.

Hedge fund assets come from HFR year-end industry reports. The realized hedge fund fee---management plus performance---is the fitted value from a regression of French's 1997--2007 realized fees on Barclay Hedge Fund Index returns ($R^{2}=0.85$), projected on returns measured relative to the index's high-water mark and floored at the 1.16\% average management fee. The realized fee is therefore cyclical: it falls to the management floor in 2008 and toward it in 2011, 2018, and 2022. The fund-of-funds share scales French's 2007 value (0.140) by the BarclayHedge fund-of-funds share, falling to 0.016 by 2024 as fund-of-funds intermediation collapsed. The US-equity share of hedge fund assets stays at French's 0.313: the equity-strategy share of hedge funds fell by a quarter while the US weight in developed equity markets rose by half, and the two trends offset.

US-equity trading revenue---commissions plus market-making gains---follows French's application of \textcite{stoll_equity_1993}. Through 2016, commissions and trading gains come from the SEC's published FOCUS aggregates. Thereafter only total broker-dealer revenue is public (FINRA), and we apply a commissions-plus-trading-gains share interpolated between the 2016 FOCUS value (0.19) and the SEC staff's published 2024 value (0.17). The US-equity fraction of this revenue is 0.4242, the mean fraction implied by French's Table V over 2002--2006; the implied fraction is trendless over 1980--2006 (range 0.33--0.55), so a constant is the appropriate form. Trading gains in 2008 are floored at zero because that year's industry-wide loss reflects fixed-income write-downs outside French's US-equity scope.

The passive benchmark follows French. The passive mutual fund cost is the expense ratio of the highest-fee live share class of the Vanguard Total Stock Market Index Fund from its SEC prospectuses, which declines from about 20 basis points in the 1990s to 14 basis points from 2017 onward. The passive institutional cost scales French's 2006 CEM passive value (2.9 basis points) by the Vanguard Institutional Index expense ratio, plus French's defined-contribution premium. The passive trading cost is the actual trading cost times the ratio of French's 10\% passive turnover to realized market turnover (annual dollar volume over market capitalization, from SIFMA and CRSP). Hedge fund assets are reallocated proportionally to the other investor types. The passive total falls from 18 basis points in 1980 to 5.5 in 2024 (Figure \ref{fig:french2008_passive_cost_breakdown}); trading costs drive the early decline, and from the 2000s onward the total is mostly the index-fund fee applied to the mutual fund allocation.

The incremental cost of active averages 46 basis points over 2007--2024. It is cyclical around the financial crisis---41 basis points in 2008, when trading gains and performance fees collapse, and 82 in 2009---and declines thereafter to 26 basis points in 2024. The composition shifts as well: trading costs are two-thirds of the active total in 1980, mutual fund fees dominate by the 2010s, and by 2024 mutual fund fees, hedge fund fees, and trading costs contribute roughly ten basis points each (Figure \ref{fig:french2008_active_cost_breakdown}). The dominant uncertainty is the US-equity fraction of trading revenue: varying it over its implied range moves the 2024 estimate between 23 and 29 basis points. Two revenue sources French could not measure remain outside our series---brokers' interest income on client cash and payment for order flow---which together would add roughly 2--5 basis points in recent years. The dollar cost of price discovery grows from \$7 billion in 1980 to \$135 billion in 2024.

\begin{figure}
\caption{\label{fig:french2008_active_cost_breakdown}The cost of active investing by component, 1980--2024}
\begin{center}
\includegraphics[width=1\columnwidth]{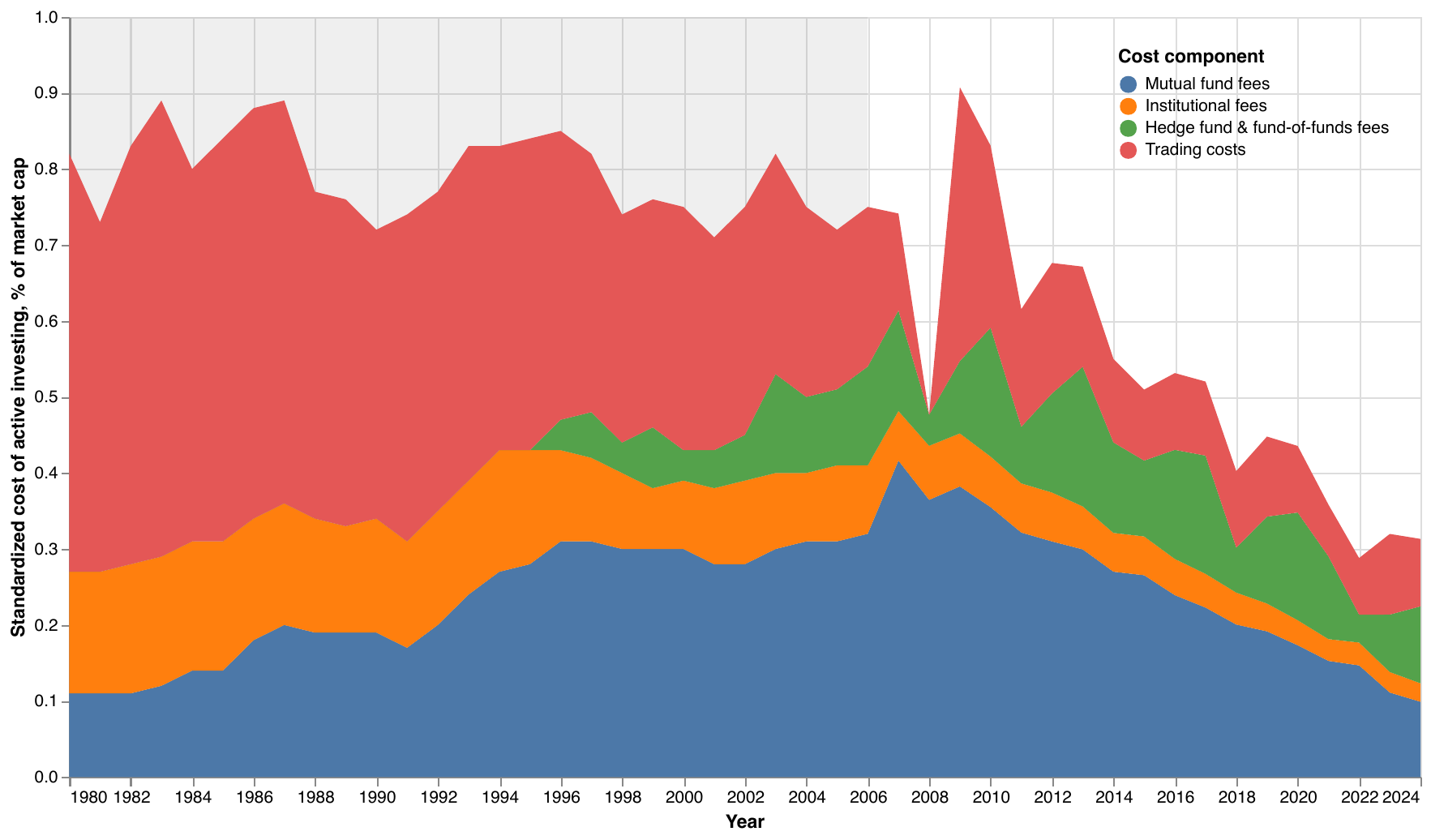}
\end{center}
\par
{\footnotesize Notes: Plotted is the standardized cost of active investing---the cost of actual investor portfolios as a percentage of US equity market capitalization---decomposed into mutual fund fees and loads, institutional investment-management fees, hedge fund and fund-of-funds fees, and US-equity trading costs. The shaded region marks French's original 1980--2006 sample.}
\end{figure}

\begin{figure}
  \caption{\label{fig:french2008_passive_cost_breakdown}The cost of passive investing by component, 1980--2024}
  \begin{center}
  \includegraphics[width=1\columnwidth]{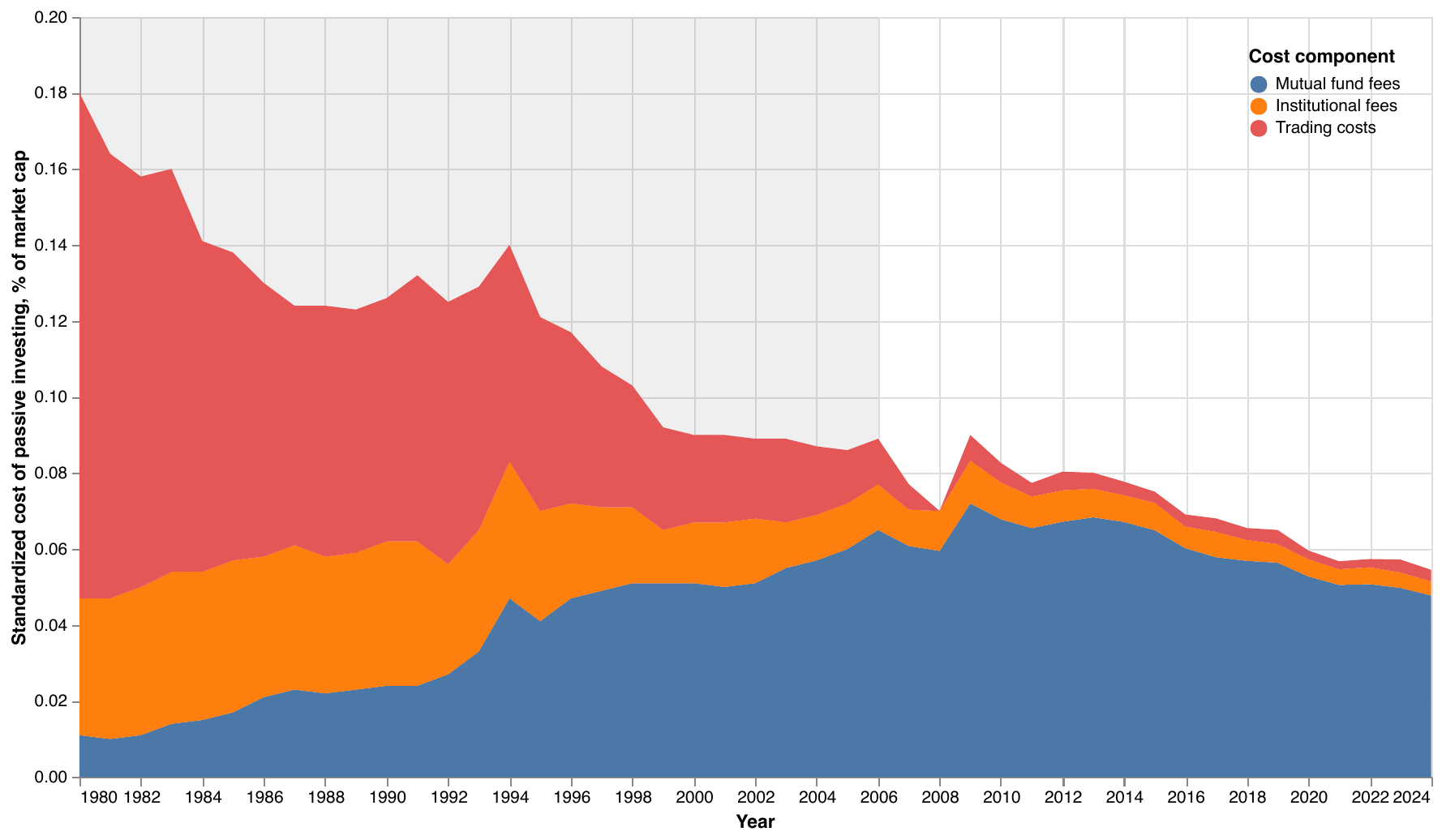}
  \end{center}
  \par
  {\footnotesize Notes: Plotted is the standardized cost of the passive benchmark portfolio as a percentage of US equity market capitalization, decomposed into mutual fund fees (the highest-fee live share class of the Vanguard Total Stock Market Index Fund), institutional passive-management costs, and trading costs at 10\% annual turnover. The passive scenario reallocates hedge fund assets proportionally, so it has no hedge fund component. The shaded region marks French's original 1980--2006 sample.}
\end{figure}
  
\end{document}

%% file: abcommonheader.tex
\usepackage{amsmath}
\usepackage{float}
\usepackage{tikz}
\usepackage{pgffor}
\usepackage{caption}
\usepackage{subcaption}
\usepackage{booktabs}
\usepackage{pdflscape}
\usepackage{xstring}

\usepackage{tgpagella}

\usepackage[T1]{fontenc}
\@ifpackageloaded{inputenc}{}{\usepackage[utf8]{inputenc}}
\usepackage{color}
\usepackage{amsthm}
\usepackage{pdflscape}
\usepackage{setspace}
\@ifpackageloaded{hyperref}{}{%
  \usepackage[unicode=true,
   bookmarks=true,bookmarksnumbered=false,bookmarksopen=false,
   breaklinks=false,pdfborder={0 0 1},backref=false,colorlinks=true]
   {hyperref}%
}
\hypersetup{pdftitle={The Value of Information: A Puzzle},
 pdfauthor={Ohad Kadan and Asaf Manela},
 linkcolor=blue!50!black,citecolor=blue!50!black}

\usepackage[textsize=tiny]{todonotes}

\usepackage{latexml} %
\iflatexml
  \usepackage[style=authoryear,
  backend=biber,
  giveninits=true,
  uniquelist=false,
  uniquename=init,
  isbn=false,
  maxcitenames=3,
  dashed=false,
  maxbibnames=999,
  doi=false,
  url=false,
  eprint=false]{biblatex}
\else
  \input{biblatex-aer}
\fi

\usetikzlibrary{shapes,arrows,positioning,chains}

\newcommand{\baseapproach}{aflam_f1_}
\newcommand{\basedir}{dCLNV_vvof_m5_t0_i1} 
\input{tables/noteableresults_\baseapproach\basedir.tex}

\newcommand{\cpibaseline}{December 2024 }
\newcommand{\inflation}{Dollar figures are CPI adjusted to \cpibaseline values. }
\newcommand{\startyyyymm}{200309}
\newcommand{\lastyyyymm}{202412}
\newcommand{\sampleperiod}{September 2003 to December 2024}
\newcommand{\dailysampleperiod}{September 11, 2003 to December 31, 2024}

\newcommand{\infoval}{The value of information is the annualized sum of the product of share price changes times signed share order flow and reported in millions of dollars. }
\newcommand{\infovallambdasigmay}{The second line provides an alternative estimator for the value of information that is annualized price impact
time order flow variance, which requires an additional assumption of linearity between price changes and signed order flow. }

\newcommand{\vofp}{Order flow volatility is the square root of the annualized one-minute signed share order flow variance. }
\newcommand{\prcimp}{Price impact is estimated by regressing one-minute log returns on contemporaneous share order flow and divided by the previous trading day's closing stock price. }
\newcommand{\size}{Size is log market equity over the previous month. }
\newcommand{\beme}{Book-to-market is book equity as of last June divided by market equity as of last December (Fama-French conventions). }
\newcommand{\momret}{Momentum is the return over the prior 2--12 months. }
\newcommand{\earnings}{Earnings indicates an earnings announcement day. }
\newcommand{\winsor}{All variables are 1\% winsorized. }
\newcommand{\shades}{Shades indicate recessions. }
\newcommand{\loginfovalterms}{The log information value equals log price impact plus log annualized order flow variance, 
so we also report the mean of log price impact (dotted line) and the mean of log annualized return variance (dashed line). }
\newcommand{\loginfovalregs}{The log information value equals log price impact plus log annualized order flow variance, 
so in Panel (B) we regress log price impact and log annualized order flow variance on the same variables. }

\newcommand{\dailysolid}{Solid line is the daily information value averaged over stocks surrounded by its 95 percent confidence interval, based on time-clustered variance estimates. }
\newcommand{\neglambda}{Observations with negative price impact are omitted. }

\newcommand{\stderror}{Standard errors clustered by date and stock are in parentheses. }
\newcommand{\stderrorci}{Standard errors clustered by date and stock are used to construct the 95\% errors bars. }
\newcommand{\stars}{* $p<0.1$, ** $p<0.05$, *** $p<0.01$. }
\newcommand{\signalgos}{three different algorithms for signing trades as buys or sells:
    \textcite[CLNV]{chakrabarty_trade_2007}, \textcite[LR]{lee_inferring_1991}, and \textcite[EMO]{ellis_accuracy_2000}}
\newcommand{\nominprice}{No minimum price shows the effect of including stocks with nominal price below \$5. }

\newcommand{\idiovar}{Idiosyncratic variance shows the effect of subtracting market variance from stock-specific variance to measure the values of idiosyncratic information. }

\newcommand*{\paramcaption}[1]{%
  \IfStrEqCase{#1}{%
    {gamma}{relative risk aversion $\gamma$}%
    {EIS}{elasticity of intertemporal substitution $1/\rho$}%
    {beta}{time discount factor $\beta$}%
    {tau}{horizon in years}%
    {wtolexp}{optimization precision}%
    {nstates}{number of states $n$}%
    {dk}{state spacing $dk$}%
    {powerUgamma}{SDF exponent $\epsilon$}%
  }[{Unknown parameter #1}]%
}
\newcommand*{\streamcaption}[1]{%
  \IfStrEqCase{#1}{%
    {}{signal every period}%
    {_single}{one-time signal}%
    {_decomp}{signal every period (decomposition)}%
    {_single_decomp}{one-time signal (decomposition)}%
    {_single_obs}{Number of stable observations}%
    {_psychic}{signal every period (psychic)}%
    {_single_psychic}{one-time signal (psychic)}%
  }[{Unknown parameter #1}]%
}

\newcommand\unfootnote[1]{%
  \begingroup
  \renewcommand\footnotemark{}
  \renewcommand\footnoterule{}
  \renewcommand\thefootnote{}\footnote{#1}%
  \addtocounter{footnote}{-1}%
  \endgroup
}

\DeclareRobustCommand{\rvdots}{%
  \vbox{
    \baselineskip4\p@\lineskiplimit\z@
    \kern-\p@
    \hbox{.}\hbox{.}\hbox{.}
  }}

\newcommand{\sym}[1]{\rlap{#1}}%

%% file: biblatex-aer.tex
\usepackage[authordate, backend=biber, uniquename=false, noibid,
doi=false,
url=false,
eprint=false]{biblatex-chicago}

\DeclareFieldFormat{citehyperref}{%
  \DeclareFieldAlias{bibhyperref}{noformat}%
  \bibhyperref{#1}}

\DeclareFieldFormat{textcitehyperref}{%
  \DeclareFieldAlias{bibhyperref}{noformat}%
  \bibhyperref{%
    #1%
    \ifbool{cbx:parens}
      {\bibcloseparen\global\boolfalse{cbx:parens}}
      {}}}

\savebibmacro{cite}
\savebibmacro{textcite}

\renewbibmacro*{cite}{%
  \printtext[citehyperref]{%
    \restorebibmacro{cite}%
    \usebibmacro{cite}}}

\renewbibmacro*{textcite}{%
  \ifboolexpr{
    ( not test {\iffieldundef{prenote}} and
      test {\ifnumequal{\value{citecount}}{1}} )
    or
    ( not test {\iffieldundef{postnote}} and
      test {\ifnumequal{\value{citecount}}{\value{citetotal}}} )
  }
    {\DeclareFieldAlias{textcitehyperref}{noformat}}
    {}%
  \printtext[textcitehyperref]{%
    \restorebibmacro{textcite}%
    \usebibmacro{textcite}}}

\AtBeginBibliography{

}

%% file: tables/noteableresults_aflam_f1_dCLNV_vvof_m5_t0_i1.tex

\newcommand{\stdinfovalpctofmarketcapw}{0.02}

\newcommand{\medianpriceimpact}{0.01}

\newcommand{\EVwannualUSDm}{3.5}
\newcommand{\medianmompct}{10}
\newcommand{\meaninfovalpctofmarketcap}{0.04}

\newcommand{\stdinfovalpctofmarketcap}{0.03}
\newcommand{\medianbemepct}{50}
\newcommand{\erneffectonvaluebothfe}{1.27}

\newcommand{\meanincrcostactivepctvoisample}{0.50}

\newcommand{\meanmktcapUSDb}{10}
\newcommand{\meanincrcostactivepctextension}{0.46}

\newcommand{\medianmktcapUSDb}{1}

\newcommand{\meanernpct}{1}

\newcommand{\meanincrcostactivepctfrenchsample}{0.67}
\newcommand{\experneffectonvaluebothfe}{3.57}

\newcommand{\medianpriceimpactpct}{1}
\newcommand{\EVwdailyUSDapprox}{13,700}

\newcommand{\negativeVwpct}{4}

\newcommand{\nstocks}{8,300}
\newcommand{\StdclusVwannualUSDm}{62}

%% file: tables/reg_aflam_f1_dCLNV_vvof_m5_t0_i1_200309_202412_lndPdYw_on_lnwtw_bemew_momretw.tex
\begin{tabular*}{\textwidth}{l@{\extracolsep{\fill}}rrrr}
\toprule
               &             \multicolumn{4}{c}{log Information value}            \\ 
\cmidrule(lr){2-5} 
               &            (1) &            (2) &            (3) &           (4) \\ 
\midrule
Size           &  0.70\sym{***} &  0.70\sym{***} &  0.59\sym{***} & 0.62\sym{***} \\ 
               &         (0.01) &         (0.01) &         (0.01) &        (0.01) \\ 
Book-to-market & -0.29\sym{***} & -0.27\sym{***} & -0.07\sym{***} &         -0.01 \\ 
               &         (0.02) &         (0.02) &         (0.01) &        (0.01) \\ 
Momentum       &  0.29\sym{***} &  0.38\sym{***} &  0.14\sym{***} & 0.21\sym{***} \\ 
               &         (0.01) &         (0.01) &         (0.01) &        (0.01) \\ 
\midrule
Date (day) FE  &                &            Yes &                &           Yes \\ 
Stock FE       &                &                &            Yes &           Yes \\ 
\midrule
$N$            &     13,855,033 &     13,855,033 &     13,854,969 &    13,855,033 \\ 
$R^2$          &           0.52 &           0.55 &           0.65 &          0.68 \\ 
Within-$R^2$   &                &           0.54 &           0.10 &          0.10 \\ 
\bottomrule
\end{tabular*}

%% file: tables/reg_aflam_f1_dCLNV_vvof_m5_t0_i1_200309_202412_lnrvofpw_lnlampw_on_lnwtw_bemew_momretw.tex
\begin{tabular*}{\textwidth}{l@{\extracolsep{\fill}}rrrrrrrr}
\toprule
               &            \multicolumn{4}{c}{log Orderflow Variance}            &                \multicolumn{4}{c}{log Price impact}               \\ 
\cmidrule(lr){2-5} \cmidrule(lr){6-9} 
               &            (1) &            (2) &            (3) &           (4) &            (5) &            (6) &            (7) &            (8) \\ 
\midrule
Size           &  1.80\sym{***} &  1.82\sym{***} &  1.53\sym{***} & 1.63\sym{***} & -1.11\sym{***} & -1.12\sym{***} & -0.95\sym{***} & -1.02\sym{***} \\ 
               &         (0.01) &         (0.01) &         (0.01) &        (0.01) &         (0.00) &         (0.00) &         (0.01) &         (0.01) \\ 
Book-to-market & -0.34\sym{***} & -0.32\sym{***} & -0.17\sym{***} & -0.04\sym{**} &  0.06\sym{***} &  0.06\sym{***} &  0.10\sym{***} &  0.03\sym{***} \\ 
               &         (0.03) &         (0.03) &         (0.02) &        (0.02) &         (0.01) &         (0.01) &         (0.01) &         (0.01) \\ 
Momentum       &  0.45\sym{***} &  0.46\sym{***} &  0.29\sym{***} & 0.27\sym{***} & -0.16\sym{***} & -0.08\sym{***} & -0.15\sym{***} & -0.06\sym{***} \\ 
               &         (0.01) &         (0.01) &         (0.01) &        (0.01) &         (0.01) &         (0.01) &         (0.01) &         (0.00) \\ 
\midrule
Date (day) FE  &                &            Yes &                &           Yes &                &            Yes &                &            Yes \\ 
Stock FE       &                &                &            Yes &           Yes &                &                &            Yes &            Yes \\ 
\midrule
$N$            &     13,855,033 &     13,855,033 &     13,854,969 &    13,855,033 &     13,855,033 &     13,855,033 &     13,854,969 &     13,855,033 \\ 
$R^2$          &           0.76 &           0.77 &           0.83 &          0.84 &           0.69 &           0.71 &           0.74 &           0.76 \\ 
Within-$R^2$   &                &           0.77 &           0.26 &          0.24 &                &           0.71 &           0.17 &           0.16 \\ 
\bottomrule
\end{tabular*}

%% file: tables/reg_aflam_f1_dCLNV_vvof_m5_t0_i1_200309_202412_lndPdYw_on_ead_lnwtw_bemew_momretw.tex
\begin{tabular*}{\textwidth}{l@{\extracolsep{\fill}}rrrr}
\toprule
               &             \multicolumn{4}{c}{log Information value}            \\ 
\cmidrule(lr){2-5} 
               &            (1) &            (2) &            (3) &           (4) \\ 
\midrule
Earnings       &  1.38\sym{***} &  1.34\sym{***} &  1.31\sym{***} & 1.27\sym{***} \\ 
               &         (0.01) &         (0.01) &         (0.01) &        (0.01) \\ 
Size           &  0.70\sym{***} &  0.70\sym{***} &  0.59\sym{***} & 0.62\sym{***} \\ 
               &         (0.01) &         (0.01) &         (0.01) &        (0.01) \\ 
Book-to-market & -0.29\sym{***} & -0.27\sym{***} & -0.07\sym{***} &         -0.01 \\ 
               &         (0.02) &         (0.02) &         (0.01) &        (0.01) \\ 
Momentum       &  0.29\sym{***} &  0.38\sym{***} &  0.14\sym{***} & 0.21\sym{***} \\ 
               &         (0.01) &         (0.01) &         (0.01) &        (0.01) \\ 
\midrule
Date (day) FE  &                &            Yes &                &           Yes \\ 
Stock FE       &                &                &            Yes &           Yes \\ 
\midrule
$N$            &     13,855,033 &     13,855,033 &     13,854,969 &    13,855,033 \\ 
$R^2$          &           0.53 &           0.56 &           0.66 &          0.69 \\ 
Within-$R^2$   &                &           0.54 &           0.12 &          0.12 \\ 
\bottomrule
\end{tabular*}

%% file: tables/reg_aflam_f1_dCLNV_vvof_m5_t0_i1_200309_202412_lnrvofpw_lnlampw_on_ead_lnwtw_bemew_momretw.tex
\begin{tabular*}{\textwidth}{l@{\extracolsep{\fill}}rrrrrrrr}
\toprule
               &            \multicolumn{4}{c}{log Orderflow Variance}            &                \multicolumn{4}{c}{log Price impact}               \\ 
\cmidrule(lr){2-5} \cmidrule(lr){6-9} 
               &            (1) &            (2) &            (3) &           (4) &            (5) &            (6) &            (7) &            (8) \\ 
\midrule
Earnings       &  1.78\sym{***} &  1.76\sym{***} &  1.67\sym{***} & 1.64\sym{***} & -0.38\sym{***} & -0.40\sym{***} & -0.33\sym{***} & -0.35\sym{***} \\ 
               &         (0.01) &         (0.01) &         (0.01) &        (0.01) &         (0.01) &         (0.01) &         (0.01) &         (0.01) \\ 
Size           &  1.79\sym{***} &  1.81\sym{***} &  1.53\sym{***} & 1.63\sym{***} & -1.11\sym{***} & -1.12\sym{***} & -0.95\sym{***} & -1.02\sym{***} \\ 
               &         (0.01) &         (0.01) &         (0.01) &        (0.01) &         (0.00) &         (0.00) &         (0.01) &         (0.01) \\ 
Book-to-market & -0.34\sym{***} & -0.32\sym{***} & -0.17\sym{***} & -0.04\sym{**} &  0.06\sym{***} &  0.06\sym{***} &  0.10\sym{***} &  0.03\sym{***} \\ 
               &         (0.03) &         (0.03) &         (0.02) &        (0.02) &         (0.01) &         (0.01) &         (0.01) &         (0.01) \\ 
Momentum       &  0.45\sym{***} &  0.46\sym{***} &  0.29\sym{***} & 0.27\sym{***} & -0.16\sym{***} & -0.08\sym{***} & -0.15\sym{***} & -0.06\sym{***} \\ 
               &         (0.01) &         (0.01) &         (0.01) &        (0.01) &         (0.01) &         (0.01) &         (0.01) &         (0.00) \\ 
\midrule
Date (day) FE  &                &            Yes &                &           Yes &                &            Yes &                &            Yes \\ 
Stock FE       &                &                &            Yes &           Yes &                &                &            Yes &            Yes \\ 
\midrule
$N$            &     13,855,033 &     13,855,033 &     13,854,969 &    13,855,033 &     13,855,033 &     13,855,033 &     13,854,969 &     13,855,033 \\ 
$R^2$          &           0.76 &           0.77 &           0.84 &          0.84 &           0.69 &           0.72 &           0.74 &           0.76 \\ 
Within-$R^2$   &                &           0.77 &           0.27 &          0.25 &                &           0.71 &           0.17 &           0.16 \\ 
\bottomrule
\end{tabular*}

%% file: tables/validation_ITI_on_lndPdYw_aflam_f1_dCLNV_vvof_m5_t0_i1_200309_201907.tex
\begin{tabular*}{\textwidth}{l@{\extracolsep{\fill}}rrrrr}
\toprule
                   & \multicolumn{1}{c}{ITI (13D)} & \multicolumn{1}{c}{ITI (impatient)} & \multicolumn{1}{c}{ITI (patient)} & \multicolumn{1}{c}{ITI (insider)} & \multicolumn{1}{c}{ITI (short)} \\ 
\cmidrule(lr){2-2} \cmidrule(lr){3-3} \cmidrule(lr){4-4} \cmidrule(lr){5-5} \cmidrule(lr){6-6} 
                   &                           (1) &                                 (2) &                               (3) &                               (4) &                             (5) \\ 
\midrule
log Info value (z) &                0.383\sym{***} &                      0.561\sym{***} &                    0.298\sym{***} &                    0.219\sym{***} &                  0.546\sym{***} \\ 
                   &                       (0.003) &                             (0.004) &                           (0.003) &                           (0.003) &                         (0.004) \\ 
\midrule
Date (day) FE      &                           Yes &                                 Yes &                               Yes &                               Yes &                             Yes \\ 
Stock FE           &                           Yes &                                 Yes &                               Yes &                               Yes &                             Yes \\ 
\midrule
$N$                &                     9,292,557 &                           9,292,557 &                         9,292,557 &                         9,013,781 &                       9,136,164 \\ 
$R^2$              &                         0.100 &                               0.188 &                             0.069 &                             0.073 &                           0.228 \\ 
Within-$R^2$       &                         0.043 &                               0.097 &                             0.026 &                             0.014 &                           0.097 \\ 
\bottomrule
\end{tabular*}

%% file: tables/validation_aggval_on_ARA_subsamples_aflam_f1_dCLNV_vvof_m5_t0_i1_20040701_20191231.tex
\begin{tabular*}{1\columnwidth}{@{\hskip\tabcolsep\extracolsep\fill}l*{5}{r}}
\toprule
 & \multicolumn{1}{c}{Full} & \multicolumn{1}{c}{2004--2008} & \multicolumn{1}{c}{2009--2013} & \multicolumn{1}{c}{2014--2019} & \multicolumn{1}{c}{Excl.\ 2008--09} \\
\midrule
\multicolumn{6}{l}{\emph{Panel A (Levels): log dPdY}} \\
ARA (1 s.d.) & $0.140\sym{***}$ & $0.178\sym{***}$ & $0.086\sym{***}$ & $0.111\sym{***}$ & $0.146\sym{***}$ \\
 & (0.009) & (0.025) & (0.025) & (0.009) & (0.009) \\
Observations & 3,900 & 1,135 & 1,255 & 1,510 & 2,513 \\
$R^2$ & 0.232 & 0.230 & 0.059 & 0.270 & 0.277 \\
\midrule
\multicolumn{6}{l}{\emph{Panel B (First differences): $\Delta$ log dPdY}} \\
ARA (1 s.d.) & $0.059\sym{***}$ & $0.070\sym{***}$ & $0.076\sym{***}$ & $0.044\sym{***}$ & $0.053\sym{***}$ \\
 & (0.004) & (0.009) & (0.007) & (0.005) & (0.004) \\
Observations & 3,899 & 1,134 & 1,254 & 1,509 & 2,512 \\
$R^2$ & 0.127 & 0.141 & 0.139 & 0.119 & 0.120 \\
\bottomrule
\end{tabular*}